\documentclass[11pt, letterpaper]{article}
\usepackage{preamble}

\newcommand{\Lloc}{\Lambda_{\mathrm{loc}}}

\newcommand{\Ttot}{T_{\mathrm{tot}}}
\newcommand{\Nexp}{N_{\mathrm{exp}}}
\newcommand{\Stab}{\mathsf{Stab}}
\newcommand{\Cl}{\mathsf{Cl}}
\newcommand{\SWAP}{\mathrm{SWAP}}

\makeatletter
\renewcommand{\theHALG@line}{\thealgorithm.\arabic{ALG@line}}
\makeatother

\definecolor{FixColor}{rgb}{0.86,0.08,0.24}

\title{Learning Arbitrary Lindbladians from Time Evolution}

\author{
  Zhili Chen\thanks{Centre for Quantum Technologies, National University of Singapore.
  \texttt{\{chen.zhili, yu.zhan\}@u.nus.edu}}%
  \and
  Zhan Yu\texorpdfstring{\footnotemark[1]}{}%
}
\date{\today}

\begin{document}
\pagenumbering{gobble}
\hypersetup{pageanchor=false}
\maketitle
\begin{abstract}
We study the problem of learning arbitrary $n$-qubit Lindbladians, the generators of Markovian open-system dynamics, from their time evolution. We give the first known efficient algorithm that learns arbitrary Lindbladians without any structural assumptions.
For an $n$-qubit Lindbladian of dynamical strength at most $\Lambda$, the algorithm estimates every coefficient to error $\varepsilon$ using $\widetilde O(\Lambda^2/\varepsilon^2)$ experiments and $\widetilde O(\Lambda/\varepsilon^2)$ total evolution time, with classical running time polynomial in $n$ and $\Lambda/\varepsilon$. Both experimental costs match known lower bounds up to polylogarithmic factors.
The algorithm identifies candidate labels through product Pauli experiments and estimates their coefficients using random stabilizer preparations and Clifford measurements. It requires no ancillas or control during the unknown evolution.

For constant $k$, our algorithm learns generic $k$-local Lindbladians of local dynamical strength at most $\Lambda_{\mathrm{loc}}$ using $\widetilde O(n^2\Lambda_{\mathrm{loc}}^2/\varepsilon^2)$ experiments and $\widetilde O(n\Lambda_{\mathrm{loc}}/\varepsilon^2)$ total evolution time, improving the upper bounds of prior works in terms of the dependence on $n$.
We further improve the lower bounds of \hyperlink{cite.ACG+26}{Arad, Chen, Guo, Rebentrost, and Yu (SODA~2027)} using a family of $3$-local Lindbladians. They already hold when the support is known, and only one Hamiltonian coefficient is requested, even for adaptive protocols with arbitrary ancillas and quantum memory. These lower bounds establish that our algorithm is optimal up to polylogarithmic factors in both experiment count and total evolution time for every fixed $k \geq 3$.
\end{abstract}

\clearpage
\tableofcontents
\clearpage
\pagenumbering{arabic}
\hypersetup{pageanchor=true}

\section{Introduction}

Learning the dynamics of a quantum system from experimental data is a basic task throughout quantum science. It underlies the calibration and certification of quantum devices~\cite{SHN+14,WGFC14,WPS+17}, quantum simulation and computation~\cite{Llo96,BSK+17}, and quantum metrology~\cite{LBS+04,DRC17}. For a closed system, the dynamics is generated by a Hamiltonian, written in the Pauli basis as $H=\sum_a h_aP_a$, and the learning task is to estimate the unknown coefficients $\{h_a\}$. Hamiltonian learning has been intensively studied under various assumptions on locality, interaction structure, temperature, evolution time, and available quantum control~\cite{BAL19,AAKS20,ZYLB21,HKT24,HTFS23,BLMT24b,BLMT24a}. The most general setting among these results is the \emph{ansatz-free} one, where the Pauli terms of the Hamiltonian may occur anywhere in the exponentially large ambient space and no interaction structure is prescribed in advance~\cite{Zha25,HMG+25,ST25,ZG26}.

Real quantum devices are rarely closed and continuously interact with their environments. Coupling to the environment produces dissipation such as relaxation, dephasing, leakage, and correlated noise, which cannot be characterized by a Hamiltonian alone. Dissipation can also be engineered as a resource for simulation~\cite{BMS+11,SA26,JGZ+26}, computation~\cite{SGAZ25,SGR+25}, and Gibbs state preparation~\cite{CKG23,BCL24,BLMT24c,DLL25,CKBG25,RFA26a,RFA26}. For time-independent Markovian open quantum systems, the evolution forms a quantum dynamical semigroup $\{e^{t\cL}\}_{t\geq0}$ generated by a \emph{Lindbladian}~\cite{GKS76,Lin76},
\[
\cL(\rho) = -\iu\sum_{P_a \neq I} h_a[P_a,\rho] + \sum_{P_a,P_b\neq I} \gamma_{ab} \ab\Big(P_a \rho P_b-\frac{1}{2}\{P_b P_a, \rho\}),
\]
where the coefficients $\{h_a\}$ encode coherent interactions and $\{\gamma_{ab}\}$ encode dissipative interactions.
The learning task for an open quantum system is therefore to estimate the Lindbladian coefficients, extending Hamiltonian learning.
Previous work has studied Lindbladian learning from steady states or real-time evolution~\cite{BGP+20,FMD+24,FMRW25,IRG+26,RIG+26,Sin26,LTW26,ACG+26,MBFR26}.

Most existing guarantees assume locality or sparsity.
Locality bounds the number of qubits on which each term acts, while geometric locality also restricts their spatial separation.
The latter can exclude crosstalk, residual couplings, or correlated noise between distant qubits~\cite{HFW20}.
Sparsity is a complementary structural promise, which allows the active terms to occur anywhere in Pauli space but assumes that at most $M$ Hamiltonian and dissipative coefficients are nonzero and that every other coefficient vanishes exactly.
This exact-zero promise can be unrealistic, since weak residual couplings or a long tail of small dissipative terms can make $M$ exponentially large even when only a few coefficients matter.
Existing methods for learning sparse Lindbladians either depend on a potentially large, instance-dependent conditioning factor~\cite{IRG+26} or require additional resources such as ancillary systems and control interleaved with the unknown evolution~\cite{RIG+26}.
These limitations motivate learning without structural assumptions.
\begin{center}
\emph{Can an arbitrary Lindbladian be efficiently learned from its time evolution?}
\end{center}

The Pauli representation contains $\Theta(16^n)$ possible Hamiltonian and dissipative coordinates, making exhaustive testing intractable.
Explicitly listing the exact support of a dense Lindbladian can also require exponential time.
However, at entrywise accuracy $\eps$, it suffices to identify all coefficients visible above the target resolution.
As the only quantitative prior information about the Lindbladian, we assume a known upper bound $\Lambda$ on the dynamical strength, which fixes the natural time scale $1/\Lambda$.
The dynamical strength promise permits exponentially many nonzero coefficients and imposes neither locality nor sparsity. We show that dynamical strength nevertheless controls the number of $\eps$-heavy coefficients, so the learning algorithm is efficient whenever the cost scales as $\poly(\Lambda, n, 1/\eps)$. The algorithmic problem is to locate the required coordinates and estimate them from the full dynamics, including all of the smaller terms. We solve this problem in time polynomial in $n$ and $\Lambda/\eps$, with experiment count and total evolution time matching known lower bounds up to polylogarithmic factors.

\subsection{Results}

\paragraph{Access model and resource measures.}
For a target Lindbladian $\cL$, we assume black-box access to the channel $e^{t\cL}$ for any time $t \geq 0$.
One \emph{experiment} prepares a quantum state, applies the unknown channel $e^{t\cL}$ once for a chosen time $t\geq0$, and measures the output.
If an algorithm performs $N$ experiments at times $t_1,\ldots,t_N$, its \emph{experiment count} is $\Nexp\coloneq N$, its \emph{total evolution time} is $\Ttot\coloneq\sum_{m=1}^{N}t_m$, and its \emph{time resolution} is the shortest positive evolution time among all experiments.
We also track its classical running time.
We call an algorithm \emph{control-free} if it does not interleave control operations with the unknown evolution, so that $e^{t\cL}$ is applied once and without interruption in each experiment.
It is \emph{ancilla-free} if the state preparation, evolution, and measurement use no entanglement with ancillary quantum systems.
These restrictions capture the \emph{in situ} learning setting considered here.

\paragraph{Strength parameter.}
We make a minimal scale assumption by supplying a known upper bound $\Lambda$ on the dynamical strength $\norm{\cL^\dag}_{\infty\to\infty}$.
This is the only quantitative prior required by the algorithm.
It bounds the rate at which observables can change and supplies the probing time scale $1/\Lambda$, without restricting locality, support, or the number of nonzero coefficients.
Any valid upper bound suffices, so standard local and normalized sparse promises also provide suitable inputs.

For a local generator, write $\cL=\sum_S\cL_S$, where $\cL_S$ collects the canonical Hamiltonian coordinates with support exactly $S$ and dissipative coordinates with union support exactly $S$.
A supplied local dynamical strength bound satisfies $\max_i\sum_{S\ni i}\norm{\cL_S^\dag}_{\infty\to\infty}\leq\Lloc$.
The triangle inequality gives
\begin{align}\label{eq:lb_global_general}
\norm{\cL^\dag}_{\infty\to\infty}\leq\sum_S\norm{\cL_S^\dag}_{\infty\to\infty}\leq n\Lloc,
\end{align}
so we may take $\Lambda=n\Lloc$.

For a normalized sparse generator, suppose at most $M$ canonical coefficients are nonzero, counting dissipative coordinates as ordered pairs, and each has magnitude at most one.
Bounding each canonical term by twice its coefficient magnitude gives
\begin{align}\label{eq:sparse_global_strength}
\norm{\cL^\dag}_{\infty\to\infty}\leq2\sum_{a\ne0}|h_a|+2\sum_{a,b\ne0}|\gamma_{ab}|\leq2M.
\end{align}
Thus $\Lambda=2M=O(M)$ is a valid supplied bound under this normalization.
The strength promise includes these local and sparse classes while also allowing dense nonlocal generators.

\paragraph{Learning arbitrary Lindbladians.}
The ambient family contains $4^n-1$ Hamiltonian coordinates and $(4^n-1)^2$ dissipative coordinates, with no prior knowledge of their locations. We therefore represent an estimate compactly by candidate sets $\hat\cS_H$ and $\hat\cS_D$ together with coefficient estimates on those sets, with every omitted coordinate reported as zero. The learning objective is formally defined as follows.

\begin{problem}[Lindbladian learning]\label{prob:lindblad_learning}
Let $\cL$ be an unknown $n$-qubit Lindbladian with dynamical strength at most $\Lambda$, whose Hamiltonian coefficients are $\{h_a\}$ and dissipative coefficients are $\{\gamma_{ab}\}$. Given black-box access to the semigroup $\{e^{t\cL}\}_{t\geq0}$, the dynamical strength bound $\Lambda$, and parameters $\eps,\delta\in(0,1)$, find estimates $\{\hat h_a\}$ and $\{\hat\gamma_{ab}\}$ such that
\[
\max_{a}{\abs[\big]{\hat h_a-h_a}} \leq \eps,
\qquad
\max_{(a,b)}{\abs[\big]{\hat\gamma_{ab}-\gamma_{ab}}} \leq \eps,
\]
with probability at least $1-\delta$.
\end{problem}

As our first main result, we solve this problem by giving an efficient, ancilla-free, and control-free algorithm for learning arbitrary Lindbladians under a known bound on dynamical strength. Its \emph{support-learning} stage uses preparations of product Pauli eigenstates and single-qubit Pauli measurements to identify a polynomial-size candidate support containing every coefficient above a chosen threshold.
Its \emph{coefficient-learning} stage uses preparations of random stabilizer states and Clifford-basis measurements to estimate all coefficients in any candidate support, with experiment count logarithmic in the support size. If $\eps>\Lambda$, it allows zero estimates without experiments. We state the experimental resource bounds in the nontrivial range $0<\eps\leq\Lambda$.

\begin{theorem}[Arbitrary Lindbladian learning, informal version of \cref{cor:learning_arbitrary_lind}]\label{cor:main_end_informal}
Let $\cL$ be an arbitrary $n$-qubit Lindbladian with dynamical strength at most $\Lambda>0$.
For every $\eps\in(0,1)$ with $\eps\leq\Lambda$, there is an ancilla-free, control-free algorithm that outputs estimates with entrywise error at most $\eps$ and success probability at least $0.99$ using
\begin{align}
\Nexp=\widetilde O(\Lambda^2/\eps^2),
\qquad \Ttot=\widetilde O(\Lambda/\eps^2).
\end{align}
On the success event, at most $9\Lambda^2/\eps^2$ estimates are nonzero.
The classical running time is $\widetilde O(n^3\Lambda^{10}\eps^{-10})$.
Each experiment uses a duration $t=\widetilde\Theta(1/\Lambda)$.
\end{theorem}

The two experimental costs in~\cref{cor:main_end_informal} match the lower bounds of~\cite{ACG+26} up to polylogarithmic factors. Those bounds use single-qubit dephasing generators, whose local and global dynamical strengths coincide, so they also apply under our global strength promise. Our algorithm uses product Pauli preparations and measurements to find candidate labels, followed by random stabilizer preparations and Clifford-basis measurements to estimate coefficients.

\paragraph{Learning generic $k$-local Lindbladians.}
A Lindbladian is $k$-local if every nonzero canonical coordinate has union support of size at most $k$.
In this setting, we consider resource costs in terms of the system size $n$ and the local dynamical strength $\Lloc$.
For fixed $k=O(1)$, all possible $k$-local Hamiltonian and dissipative coordinates form a known candidate support of size $O(n^k)$.
By \cref{eq:lb_global_general}, we may take $\Lambda=n\Lloc$ as a global strength bound.
Directly applying the coefficient learning algorithm to this candidate support then yields the following result, without requiring knowledge of the exact support.

\begin{theorem}[Generic $k$-local learning, informal version of \cref{cor:learning_klocal_lind}]\label{cor:main_klocal_informal}
Fix $k=O(1)$, and let $\cL$ be a $k$-local $n$-qubit Lindbladian with local dynamical strength at most $\Lloc>0$.
For every $\eps\in(0,1)$ with $\eps\leq n\Lloc$, there is a nonadaptive, ancilla-free, control-free algorithm that estimates all canonical coefficients to entrywise error $\eps$ with probability at least $0.99$, using
\begin{align}
\Nexp=\widetilde O(n^2\Lloc^2\log n / \eps^{2}),
\qquad \Ttot=\widetilde O(n\Lloc\log n/ \eps^{2}).
\end{align}
The classical running time is $\widetilde O(n^{k+5}\Lloc^2/\eps^2)$.
\end{theorem}

To determine whether locality alone permits lower experimental costs, we construct two $3$-local Lindbladians with identical Hamiltonian and dissipative supports. Distinguishing the sign of one Hamiltonian coefficient already requires the following resources.

\begin{theorem}[Local Lindbladian lower bounds, informal version of \cref{thm:lb_count,thm:lb_time}]\label{thm:main_lb_informal}
Let $n\geq3$, $k\geq3$, $\Lloc>0$, and $0<\eps\leq\Lloc/2048$.
Any algorithm that estimates the Hamiltonian coefficients of every $k$-local $n$-qubit Lindbladian of local dynamical strength at most $\Lloc$ to entrywise error $\eps$, with success probability at least $2/3$, has worst-case experiment count and worst-case expected total evolution time satisfying
\begin{align}
\Nexp=\Omega(n^2\Lloc^2/\eps^2),
\qquad \sup_{\cL}\E_{\cL}\Ttot=\Omega(n\Lloc/\eps^2).
\end{align}
These bounds hold for adaptive algorithms with arbitrary ancillas and quantum memory.
They already hold for a pair of $3$-local generators with identical known supports when only the Hamiltonian coefficient at a specified single-qubit Pauli operator is estimated.
\end{theorem}

For every fixed $k\geq3$, the experiment count and total evolution time in~\cref{cor:main_klocal_informal} match these lower bounds up to polylogarithmic factors. Thus, even after the support is supplied, learning one Hamiltonian coefficient can be as costly in the worst case as learning every Hamiltonian and dissipative coefficient.

\begin{remark}[Lower bounds in global strength]
By~\cref{lem:lb_strength}, the hard instances have local dynamical strength $\Theta(\Lloc)$ and global dynamical strength $\Theta(n\Lloc)$. Setting $\Lloc=\Lambda/n$ gives $\Nexp=\Omega(\Lambda^2/\eps^2)$ and $\sup_{\cL}\E_{\cL}\Ttot=\Omega(\Lambda/\eps^2)$ for $0<\eps\leq\Lambda/(2048n)$, as formalized in~\cref{cor:lb_global}. The hard instances are $3$-local, have known support, and have actual global strength between $\Lambda/8$ and $\Lambda$.
The earlier dephasing lower bounds~\cite{ACG+26} concern estimation of a dissipative coefficient. Our construction also captures the dependence on $n$ at fixed $\Lloc$ and shows that learning a Hamiltonian coefficient can incur both costs in the presence of unknown dissipation.
\end{remark}


\subsection{Related work}
\paragraph{Hamiltonian learning.}
Hamiltonian learning has been studied using individual eigenstates and Gibbs states~\cite{BAL19,QR19,AAKS20,GCC24,HKT24,BLMT24b,CAN25,LTW26}, as well as real-time evolution~\cite{EHF19,ZYLB21,GCC24,HKT24,HTFS23,DOS24,BLMT24a,MFPT24,Car24,LTW26}, under a range of assumptions concerning locality, structure, and control.
In the known-structure setting, Haah, Kothari, and Tang used convergent cluster expansions and Newton--Raphson inversion to obtain the optimal high-temperature Gibbs-state sample complexity $O\ab(\log n/(\beta\eps)^2)$ in terms of the inverse temperature $\beta$, together with a constant-resolution algorithm using $O\ab(\log n/\eps^2)$ total evolution time~\cite{HKT24}. Bakshi et~al.~extended polynomial-time Gibbs-state learning to every constant inverse temperature~\cite{BLMT24b}. Huang et~al.~used coherent-control Hamiltonian reshaping to achieve Heisenberg scaling $O\ab(\log n/\eps)$ for low-intersection Hamiltonians~\cite{HTFS23}.
For unknown local structure, Bakshi et~al.~obtained Heisenberg-limited learning with total evolution time $O\ab(r\log n/\eps)$ and time resolution $\Theta(1/r)$, both depending on the effective sparsity parameter $r$~\cite{BLMT24a}, while Lewis, Tang, and Wright later gave simpler control-free algorithms from both dynamics and high-temperature Gibbs states~\cite{LTW26}.

The setting most relevant to our work is \emph{ansatz-free} Hamiltonian learning, where $H$ is only promised to be an $M$-sparse combination of arbitrary, possibly nonlocal, Pauli strings from an exponentially large ambient family. Zhao gave the first efficient algorithm in this regime using pseudo-Choi states and recursive residual learning. The resulting total evolution time is $\widetilde O\ab(M/\eps)$ with time reversal and $\widetilde O\ab(\norm{H}_\infty^{3}/\eps^4)$ with forward-time access alone, with ancillas and coherent control in both cases~\cite{Zha25}. Hu et~al.~removed time reversal and achieved Heisenberg scaling with discrete control, using $\widetilde O\ab(M^2/\eps)$ total evolution time with ancillas or $\widetilde O\ab(M^3/\eps)$ without them~\cite{HMG+25}. Sinha and Tong subsequently improved the sparsity dependence to $\widetilde O\ab(M/\eps)$ using Bell sampling while retaining ancillas and discrete control~\cite{ST25}. The protocol of Zhou and Gong~\cite{ZG26} considers the ansatz-free setting that is closest to ours, which is control-free and ancilla-free, and uses only Pauli product state preparation and measurement. Their algorithm achieves optimal total evolution time $\Theta\ab \big( \norm{H}_\infty \log(\norm{H}_\infty / \eps)/ \eps^{2})$ with average probe time $O\ab(1/\norm{H}_\infty)$ in such a setting.

\begin{table}[htb!]
\centering
\small
\setlength{\tabcolsep}{3pt}
\renewcommand{\arraystretch}{1.3}
\begin{adjustbox}{max width=\textwidth}
\begin{tabular}{@{}llllll@{}}
\toprule
Structure & Ancilla & Additional assumptions & $\Ttot$ & $\Nexp$ & Reference\\
\midrule
\multirow{8}{*}{$k$-local}
 & No & {\renewcommand{\arraystretch}{1}\begin{tabular}[c]{@{}l@{}}Single-qubit dissipation,\\Lieb--Robinson bound, $\Lambda_{\mathrm{loc}} = O(1)$\end{tabular}} & --- & $\widetilde O\ab\big(1/\eps^{2})$ & \cite{FMD+24} \\
 & No & Bounded degree & $\widetilde O\ab\big(\Lambda_{\mathrm{loc}}/\eps^{2})$ & $\widetilde O\ab\big(\Lambda_{\mathrm{loc}}^2/\eps^{2})$ & \cite{LTW26,ACG+26,MBFR26}\\
 & No & Approximate interaction degree $d$ & $\widetilde O\ab\big(d^{2}\Lambda_{\mathrm{loc}}/\eps^{2})$ & $\widetilde O\ab\big(d^{2}\Lambda_{\mathrm{loc}}^2/\eps^{2})$ & \cite{LTW26,MBFR26}\\
 & No & None & $\widetilde O\ab\big(n^{2k-2}\Lambda_{\mathrm{loc}}/\eps^{2})$ & $\widetilde O\ab\big(n^{2k-2}\Lambda_{\mathrm{loc}}^2/\eps^{2})$ & \cite{LTW26,ACG+26,MBFR26} \\
  & No & None & $\widetilde O\ab\big(n\Lambda_{\mathrm{loc}}/\eps^{2})$ & $\widetilde O\ab\big(n^2\Lambda_{\mathrm{loc}}^2/\eps^{2})$ & This work \\
 & Arbitrary & None & $\Omega\ab\big(\Lambda_{\mathrm{loc}}/\eps^{2})$ & $\Omega\ab\big(\Lambda_{\mathrm{loc}}^2/\eps^{2})$ & \cite{ACG+26} \\
 & Arbitrary & $k\geq2$, $d=\Theta(n^{1/3})$ & $\Omega\ab\big(d^2\Lambda_{\mathrm{loc}}/(\eps^2\log d))$ & --- & \cite{LTW26} \\
 & Arbitrary & $k \geq 3$ & $\Omega\ab\big(n\Lambda_{\mathrm{loc}}/\eps^{2})$ & $\Omega\ab\big(n^2\Lambda_{\mathrm{loc}}^2/\eps^{2})$ & This work \\
\midrule
\multirow{3}{*}{$M$-sparse} & No & Coefficient gap $\Omega(\eps)$ & --- & $\widetilde O(M^4\nu^2/\eps^4)$ & \cite{IRG+26}\\
 & $n+O(\log M)$ & None & $\widetilde O(M/\eps^2+M^2/\eps)$ & --- & \cite{RIG+26}\\
 & No & None & $\widetilde O(M_0^2\Gamma/\eps^2)$ & $\widetilde O(M_0^2\Gamma^2/\eps^4)$ & \cite{ZG26a}\\
\midrule
General & No & None & $\widetilde O\ab\big(\Lambda/\eps^{2})$ & $\widetilde O\ab\big(\Lambda^2/\eps^{2})$ & This work\\
\bottomrule
\end{tabular}
\end{adjustbox}
\caption{Comparison of resource bounds for entrywise Lindbladian learning from time evolution in terms of total evolution time $\Ttot$ and experiment count $\Nexp$.
Polylogarithmic factors in $n,M,d,1/\eps,\Lambda_{\mathrm{loc}},\Lambda$ are suppressed in $\widetilde O$.
A dash indicates that the corresponding resource is not stated separately in the cited work.
Here, $k$ denotes locality and $M$ denotes sparsity.
The parameters $\Lambda_{\mathrm{loc}}$ and $\Lambda$ bound the local and global dynamical strengths, respectively.
$d$ is the approximate interaction degree used in~\cite{LTW26}, and $\nu$ is the instance-dependent conditioning factor for linear inversion in~\cite{IRG+26}.
The parameter $\Gamma \geq \Lambda$ is the global strength bound of~\cite{ZG26a}, and $M_0 \geq M$ is the sparsity budget defined in~\cite{ZG26a}.}
\label{tab:lindblad_comparison}
\end{table}

\paragraph{Lindbladian learning.}
The reconstruction of open-system dynamics has been studied using steady states, model-fitting and estimation methods for structured noise models, experimental Lindbladian tomography, and non-Markovian noise learning~\cite{Buz98,BHPC03,BGP+20,OKC23,BW24,LSK+25,BSM+26,MECT25,FMD+24,FMRW25}. The works most relevant to ours are learning from time evolution. Fran{\c c}a et~al.~combine polynomial derivative estimation with shadow process tomography for geometrically local Hamiltonians and Markovian noise obeying a Lieb--Robinson bound~\cite{FMD+24}, and later extend this framework to local time-dependent dynamics~\cite{FMRW25}.

Ivashkov et~al.~initiated ansatz-free Lindbladian learning in situ for an $M$-sparse generator using only product Pauli preparations and Pauli measurements~\cite{IRG+26}, with end-to-end experiment complexity $\widetilde O\ab(M^4\nu^2/\eps^4)$. Under the help of quantum memory, Romanov et~al.~used recursive stabilizer-code reshaping and short-time Choi observables to learn arbitrary sparse generators in total evolution time $\widetilde O\ab(M/\eps^2+M^2/\eps)$, at the cost of $n+O(\log M)$ ancillas and interleaved Clifford control~\cite{RIG+26}. Ansatz-free learning robust to state-preparation-and-measurement (SPAM) errors is studied in~\cite{Sin26}.

For local Lindbladians with low dissipative intersection, Arad et~al.~\cite{ACG+26} proposed an algorithm that achieves total evolution time $\widetilde O\ab\big(\Lambda_{\mathrm{loc}}\log n/\eps^{2})$ with a matching lower bound on $\Lambda_{\mathrm{loc}}$ and $\eps$ up to logarithmic factors. Lewis, Tang, and Wright~\cite{LTW26} learn unknown $k$-local Lindbladians of approximate interaction degree $d$\footnote{For fixed $k$, the approximate interaction degree satisfies $d\leq16^k\binom{n-1}{k-1}=O(n^{k-1})$.} using $O\ab\big(\Lloc^2d^2\log n/\eps^2)$ experiments and total evolution time $O\ab(\Lloc d^2\log n/\eps^2)$, also covering decaying interactions. M\"obus et~al.~\cite[Theorem~4.6]{MBFR26} obtain analogous bounds, up to polylogarithmic factors, with $d$ replaced by a local effective-sparsity bound $r$, provided a sufficiently small sitewise coefficient tail is promised. The approximate degree of~\cite{LTW26} implies their sparsity and tail conditions with $r=5d$~\cite[Remark~9.1]{MBFR26}.
Without degree promises, all three works use $\widetilde O\ab\big(n^{2k-2}\Lloc^2/\eps^2)$ experiments and total evolution time $\widetilde O\ab\big(n^{2k-2}\Lloc/\eps^2)$ to learn generic $k$-local Lindbladians. M\"obus et~al.~\cite[Theorem~5.3]{MBFR26} also recover a valid generator to diamond-norm error $\eps$ using $\widetilde O\ab\big(n^{2k}\Lloc^2/\eps^2)$ experiments.
Lewis, Tang, and Wright~\cite[Theorem~6.1]{LTW26} further prove a total evolution time lower bound $\Omega\ab\big(d^2\Lloc/(\eps^2\log d))$ for entrywise learning of $2$-local Lindbladians, even with adaptive queries and arbitrary ancillas. Their hard instance has $d=\Theta(n^{1/3})$, so the bound is $\Omega\ab\big(n^{2/3}\Lloc/(\eps^2\log n))$ in terms of system size.

Our algorithm applies to arbitrary Lindbladian without assumptions of locality, sparsity, known support, or a coefficient gap, using ancilla-free and control-free experiments.
Its experiment count and total evolution time are optimal up to polylogarithmic factors, matching the lower bounds of~\cite{ACG+26}.
For learning generic $k$-local Lindbladians with $\Lloc$ as input and $k \geq 2$, our algorithm in \Cref{cor:main_klocal_informal} outperforms prior results~\cite{ACG+26,LTW26,MBFR26}.
Our coefficient-learning stage uses global Clifford preparations and measurements, which may entangle the system qubits.
The protocols of~\cite{LTW26,ACG+26,MBFR26} use product Pauli preparations and measurements.
Our matching lower bounds in~\cref{thm:main_lb_informal} establish the near-optimal dependence on $n$ and $\Lloc$ for this result.
Our algorithm also improves on the bounds of~\cite{LTW26,MBFR26} under an approximate interaction degree promise, with less total evolution time for $d=\omega(\sqrt n)$ and fewer experiments for $d=\omega(n)$.
A detailed comparison is provided in~\cref{tab:lindblad_comparison}.

\paragraph{Independent and concurrent work.}
Zhou and Gong~\cite{ZG26a} independently developed a protocol for learning sparse Lindbladians using product Pauli preparations and measurements throughout. Given a sparsity budget $M_0 \geq M$ and a strength bound $\Gamma \geq \Lambda$ as defined in their work, their algorithm uses $\widetilde O\ab(M_0^2\Gamma^2/\eps^4)$ experiments and $\widetilde O\ab(M_0^2\Gamma/\eps^2)$ total evolution time. It also provides robustness guarantees under calibrated Pauli-diagonal SPAM errors. Both protocols use observed Pauli displacements to identify candidate labels. Our support guarantee applies to arbitrary Lindbladians under a dynamical strength bound alone, and our coefficient-learning stage uses global Clifford probes.

In independent concurrent work on quantum process tomography, Bravo-Prieto, Gong, and Mele~\cite{BGM26} established dimension-independent directional second-moment bounds for an
ancilla-free Choi estimator using Haar-random inputs and output bases. We independently establish such a bound for the Bell-basis witnesses used in generator learning, using stabilizer preparations and Clifford-basis measurements. Each witness estimator can be evaluated efficiently from a pair of stabilizer tableaux. Combining these estimators with endpoint differentiation allows the same collection of experimental data to estimate all
requested Lindbladian coefficients.

\subsection{Technical overview}

In our proposed algorithm, both learning stages use Bell-basis matrix elements of the time-evolution Choi state.
The endpoint derivatives of these elements encode the coefficients of $\cL$.
For $d=2^n$, let $\ket{\Phi_0} \coloneq \frac{1}{\sqrt d}\sum_{j=0}^{d-1}\ket{j}\ket{j}$ be the maximally entangled state, and let $\ket{\Phi_a} \coloneq (I\otimes P_a)\ket{\Phi_0}$.
Denote the Choi state of the time-$t$ channel by
\[
J_t\coloneq
(\mathrm{id}\otimes e^{t\cL})\ab\big(\ketbra{\Phi_0}{\Phi_0}).
\]
The channel $\chi$-matrix, defined by $e^{t\cL}(\rho)=\sum_{a,b}\chi_{ab}(t)P_a\rho P_b$, is exactly the Bell-basis matrix of $J_t$. Its endpoint derivative is the $\chi$-matrix of the Lindbladian, defined by $\cL(\rho)=\sum_{a,b}\chi_{ab}P_a\rho P_b$:
\[
g_{ab}(t)\coloneq \bra{\Phi_a}\,J_t\,\ket{\Phi_b}=\chi_{ab}(t),
\qquad
g'_{ab}(0) \coloneq \bra{\Phi_a}\,J'_0\,\ket{\Phi_b}=\chi_{ab}.
\]
The correspondence between the coefficients and the $\chi$-matrix of the Lindbladian gives
\[
h_a = \Imag\bra{\Phi_0}\, J'_0 \, \ket{\Phi_a} = \Imag g'_{0a}(0), \qquad \gamma_{ab}=\bra{\Phi_a}\, J'_0\, \ket{\Phi_b} = g'_{ab}(0).
\]
The algorithms do not physically prepare $J_t$.
The support-learning stage uses positivity of $J_t$ to convert large coefficients into observable Pauli-error events, while the coefficient-learning stage estimates selected matrix elements of $J_t$ through system-only classical shadows.

\paragraph{Endpoint differentiation at the standard quantum limit.}
In both stages, the quantity of interest is the derivative at $t=0$ of a signal that can be sampled only at positive times.
The signal is a Bell coherence $g_{ab}(t)$ of the Choi state in support learning and a witness signal $f_A(t)$ in coefficient learning.
The naive finite difference $[f(\tau)-f(0)]/\tau$ has bias $O(\Lambda^2\tau)$, which forces $\tau=O(\eps/\Lambda^2)$.
Dividing the statistical error by such a small $\tau$ inflates the sample complexity to $O(1/\eps^4)$~\cite{BGP+20,ZYLB21}.
Following the interpolation technique of recent generator-learning algorithms~\cite{Car24,FMD+24,GCC24,HLS+26,IRG+26,ACG+26}, we instead place $q+1$ Chebyshev--Lobatto nodes on a short interval $[0,T]$ and differentiate the degree-$q$ interpolating polynomial at the endpoint $t=0$.
This gives $f'(0)\approx\sum_jw_jf(t_j)$, where the $w_j$ are precomputed weights.
The target signals obey the uniform derivative bounds $\abs[\big]{g^{(r)}_{ab}(t)}\leq\Lambda^r$ and $\abs[\big]{f^{(r)}_A(t)}\leq2\Lambda^r$, so fixing $T=1/\Lambda$ makes the interpolation bias decay as $\Lambda/(q+1)!$, and $q=\widetilde\Theta(1)$ nodes suffice to achieve bias $\eps/2$.
The positive-time weights have total absolute weight $W=\sum_{j>0}\abs{w_j}=O(q^2\Lambda)$.
With a constant-variance estimator for the finite-time signal, estimating one derivative to accuracy $\eta$ costs $\widetilde O(\Lambda^2/\eta^2)$ experiments and total evolution time at most $\widetilde O(\Lambda/\eta^2)$.
Support learning uses this rule in reverse to certify finite-time population, whereas coefficient learning uses it directly to estimate derivatives.

\paragraph{Support learning by displacement sampling.}
Let $\lambda_a(t)=g_{aa}(t)$. Since $J_t$ is a quantum state, the $\{\lambda_a(t)\}_a$ form a probability distribution, and positivity gives $\abs{g_{ab}(t)}^2\leq\lambda_a(t)\lambda_b(t)$. Normalize the positive-time endpoint weights as $p_j=\abs{w_j}/W$ and define the Chebyshev mixture $\mu(a)=\sum_{j=1}^qp_j\lambda_a(t_j)$. Applying the endpoint estimate contrapositively shows
\[
\abs{g'_{ab}(0)}\geq\eta
\quad\Longrightarrow\quad
\mu(a),\mu(b)\geq
\alpha\coloneq\frac{\eta^2}{4W^2}
=\widetilde\Omega\ab\Big(\frac{\eta^2}{\Lambda^2}).
\]
Thus $\abs{\gamma_{ab}}\geq\eta$ makes both $a$ and $b$ $\alpha$-heavy labels in $\mu$, and $\abs{h_a}\geq\eta$ makes $a$ an $\alpha$-heavy label in $\mu$.
Because $\mu$ is a probability distribution, at most $1/\alpha=\widetilde O(\Lambda^2/\eta^2)$ labels are $\alpha$-heavy.
A direct norm bound in~\cref{lem:effective_sparsity} further gives $2\abs{\cS_H^{\geq\eta}}+\abs{\cS_D^{\geq\eta}}\leq\Lambda^2/\eta^2$.
The marginal-sampling procedure below constructs a larger candidate family containing these heavy coordinates.

The distribution $\mu$ can be sampled physically without measuring a Bell coherence. Sampling $j$ according to $p_j$ and Pauli-twirling the corresponding evolution produces the Pauli channel
\[
\closure{\cE}(\rho)
=\sum_{j=1}^q p_j\frac{1}{4^n}\sum_r
P_r e^{t_j\cL}(P_r\rho P_r)P_r
=\sum_a \mu(a)P_a\rho P_a.
\]
For $c=(z,x)$, a computational-basis experiment reveals the displacement $x$, while an $X$-basis experiment reveals $z$. Sampling these two marginals yields lists containing the $x$- and $z$-components of every heavy label. Their Cartesian product gives a candidate set $\hat\cS$, from which we return $\hat\cS_H=\hat\cS$ and $\hat\cS_D=\hat\cS\times\hat\cS$. False combinations are harmless because the next stage estimates their coefficients.

\paragraph{Coefficient witnesses on the virtual Choi state.}
Given a candidate family, polarization reduces its off-diagonal Bell entries to diagonal rank-one functionals. For a normalized operator $A$, let $\ket{\Phi_A}=(I\otimes A)\ket{\Phi_0}$ and $f_A(t)=\bra{\Phi_A}J_t\ket{\Phi_A}$. Together with the diagonal witnesses $P_a$, we use
\begin{alignat*}{2}
A_{ab}& = \frac{P_a+P_b}{\sqrt2},\quad
& f_{A_{ab}}(t)&=\frac{f_{P_a}(t)+f_{P_b}(t)}{2}+\Real g_{ab}(t),\\
A^{\iu}_{ab}& =\frac{P_a+\iu P_b}{\sqrt2},
& f_{A^{\iu}_{ab}}(t)&=\frac{f_{P_a}(t)+f_{P_b}(t)}{2}-\Imag g_{ab}(t).
\end{alignat*}
Differentiating these identities recovers $\gamma_{ab}$, whereas the imaginary witness for the boundary pair $(a,0)$ recovers $h_a$~(\cref{lem:witness_signals,lem:witness_identities}).
Each coefficient uses at most four witness derivatives, so a candidate family of size $M$ generates at most $4M+1$ witnesses.

\paragraph{Ancilla-free process shadows with constant variance.}
The witness signals appear to require the doubled-space state $J_t$, but they can be estimated on the original system alone. Each experiment prepares a uniformly random stabilizer state $\psi$, evolves it under $e^{t\cL}$, measures in an independently random Clifford basis, and records the resulting stabilizer projector $\phi$~\cite{HKP20,LLC24}. The associated classical snapshot satisfies
\[
\hat J_t =\ab\big((d+1)\psi^\top - I) \otimes \ab\big((d+1)\phi-I),
\qquad \E[\hat J_t] =J_t.
\]
A similar estimator was independently developed in concurrent work on quantum process tomography~\cite{BGM26}, but with Haar-random inputs and Haar-random output bases in place of our Clifford ensembles. We then construct the witness estimator $X_A = \tr\ab\big(\ketbra{\Phi_A}{\Phi_A}\hat J_t)$, which satisfies 
\[
\E[X_A] = f_A(t)
\qquad\text{and}\qquad
\E[X_A^2] \leq 160.
\]
The Born weighting contributes one extra stabilizer projector, so the third moment of the uniform stabilizer ensemble controls the second moment of $X_A$.
Median-of-means estimation and a union bound give simultaneous accuracy for all witnesses, with the candidate size entering only through $\log M$.
Each snapshot value is evaluated from stabilizer tableaux in $O(n^3)$ time~\cite{AG04}.

Finally, we run support learning at threshold $\eta=\eps$, estimate every candidate coefficient to error $\eps/3$, and output only estimates of magnitude at least $2\eps/3$, assigning zero elsewhere.
On the success event, the final estimate has entrywise error at most $\eps$ and $O(\Lambda^2/\eps^2)$ nonzero coordinates.
The algorithm uses $\widetilde O(\Lambda^2\eps^{-2}\log(1/\delta))$ experiments and $\widetilde O(\Lambda\eps^{-2}\log(1/\delta))$ total evolution time, with polynomial classical running time.

\paragraph{Lower bounds from a family of $3$-local Lindbladians.}
Lewis, Tang, and Wright~\cite{LTW26} construct $2$-local hard instances on a star. Each jump is a linear combination of $Z_i$ and $Z_0Z_i$, coupling a central qubit $0$ to one bath qubit $i$. They choose the cross terms in the dissipators to cancel the direct signal of the Hamiltonian coefficient at $Z_0$. To hide the remaining signal, they use random sparse couplings with matching higher moments, making two ensembles with different signs of that coefficient hard to distinguish. This yields a degree-dependent lower bound on total evolution time.

Our construction yields lower bounds matching the upper bounds in their dependence on $n$ for both experiment count and total evolution time. The quadratic experiment-count bound relies on suppressing the single-index bath terms that remain after canceling the direct Hamiltonian signal. We achieve this suppression by adding jumps involving pairs of bath qubits. Specifically, we construct the $3$-local pair $\cL_{\pm h}$ in~\cref{eq:lb_family}, with $h=2\eps$ and canonical Hamiltonians $\pm hZ_0$. On the central qubit and $m=n-1$ bath qubits, we use jumps
\[
J_Q(h)=\sqrt{r_Q}Q+\iu\frac{u_Q}{\sqrt{r_Q}}Z_0Q,
\qquad Q\in\{Z_i:i\in[m]\}\cup\{Z_iZ_j:i<j\}.
\]
We set $r_i=r$, $r_{ij}=r/(m-1)$, with $r=\Lloc/16$, and choose $u_i=2h/m$ and $u_{ij}=-2h/[m(m-1)]$. These couplings satisfy
\[
\sum_Q u_Q=h,
\qquad u_i+\sum_{j\ne i}u_{ij}=0.
\]
The first identity cancels the direct Hamiltonian signal through the cross terms in the dissipators. The second removes all single-index terms from the sign-dependent part of the bath generator, leaving only pair terms. With only $Q=Z_i$, the second condition would force every $u_i=0$, contradicting the first for $h\ne0$. The additional pair jumps make the two cancellations compatible. Orthogonality of the remaining pair contributions gives a $1/[m(m-1)]$ factor in the squared-error estimate, yielding the quadratic dependence on $n$ in the experiment-count bound. We also add central dephasing at rate $r$, which suppresses the coherence carrying the sign information at long evolution times.

Together, these cancellations limit how well the two evolutions can be distinguished in a single experiment, regardless of its duration or input state. We quantify this through the overlap between purified outputs of the two channels. Specifically, \cref{lem:lb_uniform_overlap} provides, for every $t\geq0$, Stinespring isometries $V_{\pm,t}$ mapping the system into the joint space of the system and a common auxiliary environment $E$, with
\[
e^{t\cL_{\pm h}}(\rho)=\tr_E\ab\big(V_{\pm,t}\rho V_{\pm,t}^\dag),
\qquad V_{\pm,t}^\dag V_{\pm,t}=I.
\]
These isometries can be chosen so that their overlap is a scalar multiple of the identity:
\[
\begin{gathered}
V_{+,t}^\dag V_{-,t}=s(t)I,
\qquad 0\leq1-s(t)\leq\min\{\beta,\kappa t\},\\
\beta=O\ab\Big(\frac{\eps^2}{n^2\Lloc^2}),
\qquad \kappa=O\ab\Big(\frac{\eps^2}{n\Lloc}).
\end{gathered}
\]
For any pure input, including one entangled with arbitrary ancillas, the two purified outputs therefore have overlap $s(t)$. The bound $\beta$ limits the loss of overlap per experiment, while $\kappa t$ bounds that loss in proportion to the evolution time. Accumulating these bounds over an adaptive protocol shows that constant distinguishing advantage requires $\Nexp=\Omega(1/\beta)$ and $\max_{d\in\{+,-\}}\E_d\Ttot=\Omega(1/\kappa)$, proving both resource lower bounds in~\cref{thm:main_lb_informal}. Since the proof uses a fixed pair with identical supports, these lower bounds also hold when the support is known.

\subsection{Discussion}
Our upper bounds, together with the lower bounds of~\cite{ACG+26} and our lower bounds for local Lindbladians, clarify how dissipation, quantum memory, and locality affect the cost of learning quantum dynamics.


\paragraph{Hamiltonian versus Lindbladian learning in situ.}
For control-free and ancilla-free Hamiltonian learning, the optimal total evolution time is $\Theta(\Lambda\log(\Lambda/\eps)/\eps^{2})$~\cite{ZG26}.
Our algorithm learns every Hamiltonian and dissipative coefficient of an arbitrary Lindbladian with total evolution time $\widetilde O(\Lambda/\eps^{2})$, which matches the lower bound of~\cite{ACG+26} up to polylogarithmic factors.
Passing from closed to open dynamics therefore costs only polylogarithmic overhead, although the dissipative coordinates are quadratically more numerous than the Hamiltonian ones and each probed signal mixes many of them at once.
With coherent control, Hamiltonian learning achieves the Heisenberg scaling $1/\eps$~\cite{HTFS23,BLMT24a,HMG+25,ST25}.
For Lindbladian learning, the lower bound of~\cite{ACG+26} enforces the standard quantum limit $1/\eps^{2}$, even for fully adaptive protocols with arbitrary ancillas.

\paragraph{Quantum memory does not help.}
Our coefficient-learning procedure estimates the required Bell matrix elements from classical snapshots generated by random stabilizer inputs and Clifford-basis measurements, without physically preparing the Choi state.
The resulting ancilla-free learning algorithms attain experiment counts and total evolution times that match lower bounds allowing adaptive protocols with arbitrary quantum memory and entangled measurements~\cite{ACG+26}, up to polylogarithmic factors.
Thus, quantum memory cannot improve the worst-case polynomial scaling of either resource for the stated entrywise learning task under a dynamical-strength bound.
The situation differs in closely related tasks, such as state tomography, process tomography, and Pauli-channel learning, where entangling the system with a quantum memory gives up to exponential advantages~\cite{ACQ22,CZSJ22,CCHL22,HBC+22,CLL24,COZ+24,CG25,BGM26}.
Within Lindbladian learning, the QEC-based sparse algorithm of~\cite{RIG+26} uses $n+O(\log M)$ ancillas together with Clifford control interleaved with the evolution.
Whether memory or control helps for restricted variants of the problem, such as learning the Hamiltonian part alone under weak dissipation, remains open.

\paragraph{Locality alone does not help.}
Locality reduces the number of possible coefficients to $O(n^k)$ for fixed $k$ and allows the support-learning stage to be bypassed. Nevertheless, under the same local dynamical strength bound $\Lloc$, arbitrary Lindbladians and $k$-local Lindbladians have the same optimal experimental costs up to polylogarithmic factors for every fixed $k\geq3$. Indeed, by Corollary~\ref{cor:main_klocal_informal}, our algorithm for arbitrary Lindbladians uses $\widetilde O(n^2\Lloc^2/\eps^2)$ experiments and $\widetilde O(n\Lloc/\eps^2)$ total evolution time. Our $3$-local construction in~\cref{thm:main_lb_informal} yields matching lower bounds, even when the support is known and only one Hamiltonian coefficient needs to be estimated. The difficulty thus persists after support learning. Dissipative interactions involving at most three qubits can already obscure a local Hamiltonian signal enough to incur both costs. This conclusion concerns Lindbladian learning with locality alone. A bound on the size of each interaction still permits many interactions to involve the same qubit. Stronger promises, such as bounded interaction degree, can yield lower experimental costs~\cite{LTW26,MBFR26}.

\section{Background}\label{sec:background}

\subsection{Notation}
For $n\in\N$, we write $[n]\coloneq\{1,\dots,n\}$. We denote by $\F_2=\{0,1\}$ the binary field and by $\F_2^{2n}$ the $2n$-dimensional vector space over $\F_2$. Vector addition in $\F_2^{2n}$ (bitwise XOR) is written $\oplus$. For a finite set $S$, $\abs{S}$ denotes its cardinality. The indicator $\mathbf{1}[E]$ equals $1$ if the predicate $E$ holds and $0$ otherwise. We use the standard asymptotic notation $O(\cdot)$, $\Omega(\cdot)$, and $\Theta(\cdot)$ and write $\widetilde O(\cdot)$, $\widetilde\Omega(\cdot)$, and $\widetilde\Theta(\cdot)$ when suppressing polylogarithmic factors.

All Hilbert spaces are finite-dimensional, and all matrix norms are Schatten norms. For a linear operator $A$ on a Hilbert space $\mathscr{H}$, $\norm{A}_{p}$ denotes its Schatten-$p$ norm. In particular, $\norm{A}_{1}$ is the trace norm, $\norm{A}_2$ is the Hilbert--Schmidt norm, and $\norm{A}_{\infty}$ is the operator norm. The set $\cB(\mathscr{H})$ of linear operators on $\mathscr{H}$ is itself a Hilbert space under the Hilbert--Schmidt inner product $\langle A,B\rangle\coloneq\tr\ab(A^\dagger B)$. A Hermitian operator $A$ is positive semidefinite, written $A\succeq0$, if all its eigenvalues are nonnegative.

A linear map $\cN\colon\cB(\mathscr{H})\to\cB(\mathscr{H})$ is called a \emph{superoperator}. The induced operator norm of a superoperator is
\[
  \norm{\cN}_{\infty\to\infty}
  \coloneq \sup_{A\neq 0}\frac{\norm{\cN(A)}_{\infty}}{\norm{A}_{\infty}},
\]
and its adjoint $\cN^\dag$ is defined with respect to the Hilbert--Schmidt inner product by $\langle A,\cN(B)\rangle=\langle\cN^\dag(A),B\rangle$ for all $A,B\in\cB(\mathscr{H})$. In particular, for a Hermitian operator $A$ and a Hermiticity-preserving map $\cN$,
\[
  \tr\ab\big(A\,\cN(B))=\tr\ab\big(\cN^\dag(A)\,B),
\]
which is the usual duality between the Schr\"odinger picture (evolution of states) and the Heisenberg picture (evolution of observables).

\subsection{Pauli operators}
We work throughout in the basis of tensor products of Pauli matrices.

\begin{definition}[Pauli matrices]
    The Pauli matrices are the following $2\times2$ Hermitian matrices:
\[
    I = \begin{pmatrix} 1 & 0 \\ 0 & 1 \end{pmatrix}, \quad
    X = \begin{pmatrix} 0 & 1 \\ 1 & 0 \end{pmatrix}, \quad
    Y = \begin{pmatrix} 0 & -\iu \\ \iu & 0 \end{pmatrix}, \quad
    Z = \begin{pmatrix} 1 & 0 \\ 0 & -1 \end{pmatrix}.
\]
\end{definition}

These matrices are unitary and Hermitian, and the nonidentity Pauli matrices are traceless. An $n$-qubit Pauli operator is a tensor product $P_1\otimes\cdots\otimes P_n$ with $P_i\in\{I,X,Y,Z\}$ for all $i\in[n]$, and we denote by $\cP_n$ the set of all $4^n$ $n$-qubit Pauli operators. For ease of notation and computation, we encode each Pauli operator by a binary vector, so that operator multiplication, up to a phase, becomes addition modulo two.

\begin{definition}[Binary representation of Pauli operators]\label{def:binary_rep}
    Given a Pauli operator $P_a$, its \emph{binary representation} is the vector $a=(z,x)\in\F_2^{2n}$ such that
    \[
        P_a=\bigotimes_{i=1}^n (P_a)_i=\bigotimes_{i=1}^n (-\iu)^{z_i \cdot x_i} Z^{z_i}X^{x_i},
    \]
    where each $z_i,x_i\in \{0,1\}$.
\end{definition}
The correspondence $a\mapsto P_a$ is a bijection between $\F_2^{2n}$ and $\cP_n$, and the zero vector $0\in\F_2^{2n}$ labels the identity. We refer to elements of $\F_2^{2n}$ as \emph{Pauli labels} and identify $P_a$ with its label $a$ when no confusion can arise. We write $(P_a)_i\in\{I,X,Y,Z\}$ for the single-qubit Pauli factor of $P_a$ acting on the $i$th qubit. Two basic quantities associated with a Pauli operator are its support and weight, which record the qubits on which it acts nontrivially and their number, respectively.

\begin{definition}[Support and weight of a Pauli operator]\label{def:pauli_support}
    For a Pauli operator $P_a \in \cP_n$, its \emph{support} $\supp(a)~\subseteq~[n]$ is the subset of qubits on which $P_a$ acts nontrivially, and its \emph{weight} is the size of the support:
    \[
        \supp(a)\coloneq \{i\in[n]:(P_a)_i\neq I\},
        \qquad
        \wt(a)\coloneq \abs{\supp(a)}.
    \]
    For a pair of Pauli labels $(a,b)$, we set
    \[
      \supp(a,b)\coloneq \supp(a)\cup\supp(b),
      \qquad
      \wt(a,b)\coloneq \abs{\supp(a,b)}.
    \]
\end{definition}

The multiplication and commutation structure of $\cP_n$ is captured at the level of labels by the symplectic inner product.

\begin{definition}[Symplectic inner product]\label{def:sinprod}
    Given two Pauli labels $a=(z,x)$ and $b=(z',x')$ in $\F_2^{2n}$, their symplectic inner product is
    \[
      \sinprod{a}{b} \coloneq z\cdot x' + x\cdot z' \pmod 2.
    \]
\end{definition}

The Lindbladian generator is built from commutators and anticommutators, which we recall next.
\begin{definition}[Commutator and anticommutator]
    Given operators $A$ and $B$, the commutator of $A$ and $B$ is defined as
    \[
        [A,B] = AB - BA,
    \]
    and the anticommutator is defined as
    \[
        \{A,B\} = AB + BA.
    \]
\end{definition}

Direct computation from \cref{def:binary_rep} yields the following standard properties, which we use throughout. In particular, the symplectic inner product determines whether two Pauli operators commute. For all $a,b\in\F_2^{2n}$,
\begin{equation}\label{eq:pauli_commutation}
    P_aP_b=(-1)^{\sinprod{a}{b}}P_bP_a,
\end{equation}
so $P_a$ and $P_b$ commute if $\sinprod{a}{b}=0$ and anticommute if $\sinprod{a}{b}=1$. Since Pauli operators square to the identity, \cref{eq:pauli_commutation} also gives the conjugation rule
\begin{equation}\label{eq:pauli_conjugation}
    P_bP_aP_b=(-1)^{\sinprod{a}{b}}P_a.
\end{equation}

Pauli operators are closed under multiplication up to a phase.
\begin{fact}\label{fact:xi_ab}
Given $P_a,P_b \in \cP_n$, their product has the form
\[
  P_aP_b = \xi_{ab}P_{a\oplus b},
\]
where the phase $\xi_{ab}\in\{\pm 1,\pm \iu\}$ is computable in time $O(n)$ given $a$ and $b$. Moreover, $\xi_{aa}=1$, $\xi_{ab}\xi_{ba}=1$, and $\xi_{ab} \in \{\pm 1\}$ if $P_a$ and $P_b$ commute, while $\xi_{ab} \in \{\pm \iu\}$ if they anticommute.
\end{fact}
Distinct Pauli operators are orthogonal under the Hilbert--Schmidt inner product.
\begin{fact}[Orthogonality of Pauli operators]\label{fact:pauli_inprod}
    Given $P_a,P_b \in \cP_n$, we have
    \[
        \tr\ab(P_aP_b) = 2^n\, \mathbf{1}[a=b].
    \]
    Consequently, $\{2^{-n/2}P_a\}_{a\in\F_2^{2n}}$ is an orthonormal basis of $\cB(\mathscr{H})$: every $A\in\cB(\mathscr{H})$ has the unique expansion $A=2^{-n}\sum_{a}\tr\ab(P_aA)\,P_a$.
\end{fact}
We record the character-sum orthogonality of the symplectic form, which underlies properties of the Walsh--Hadamard transform.
\begin{fact}[Symplectic character orthogonality]\label{fact:char_orth}
    For every $a\in\F_2^{2n}$,
    \[
        \frac{1}{4^n}\sum_{w\in\F_2^{2n}}(-1)^{\sinprod{a}{w}}=\mathbf{1}[a=0].
    \]
\end{fact}
\begin{proof}
    For $a=0$ every summand is $1$. For $a=(z,x)\neq0$, the map $w\mapsto\sinprod{a}{w}$ is a nonzero $\F_2$-linear functional on $\F_2^{2n}$: if $z_i=1$ for some $i$, the label $w=(0,e_i)$ satisfies $\sinprod{a}{w}=1$, and symmetrically if $x_i=1$ for some $i$. A nonzero linear functional takes each of the values $0$ and $1$ on exactly half of $\F_2^{2n}$, so the sum vanishes.
\end{proof}

\subsection{Markovian open quantum systems}

An $n$-qubit quantum system is described by the Hilbert space $\mathscr{H}=(\bbC^2)^{\otimes n}$ of dimension $2^n$.
A quantum state is a density operator $\rho\in\cB(\mathscr{H})$ satisfying $\rho\succeq0$ and $\tr(\rho)=1$.
A quantum channel is a completely positive and trace-preserving (CPTP) superoperator $\cE\colon\cB(\mathscr{H})\to\cB(\mathscr{H})$.

A quantum dynamical semigroup is a norm-continuous family of channels $\{\cE_t\}_{t\geq0}$ satisfying $\cE_0=\mathrm{id}$ and $\cE_{t_1+t_2}=\cE_{t_1}\circ\cE_{t_2}$. A Markovian open quantum system is described by such a semigroup. In finite dimensions, the Gorini--Kossakowski--Sudarshan--Lindblad (GKSL) theorem~\cite{GKS76,Lin76} characterizes the general form of the generator of these dynamics.

\begin{definition}[Lindbladian]
    A \emph{Lindbladian} is the generator of a quantum dynamical semigroup $\{\cE_t\}_{t\geq0}$ that can be written in the Pauli basis as
\begin{equation}\label{eq:Lindbladian_pauli}
      \cL(\rho) = \sum_{a\neq 0} -\iu h_a[P_a,\rho] + \sum_{a,b \neq 0} \gamma_{ab} \ab\Big(P_a \rho P_b-\frac{1}{2}\{P_b P_a,\rho\}), \qquad h_a \in \bbR \text{ and }\gamma_{ab} \in \bbC,
\end{equation}
where the \emph{Kossakowski matrix} $\gamma\coloneq(\gamma_{ab})_{a,b\neq0}\in\bbC^{(4^n-1)\times(4^n-1)}$ is Hermitian and positive semidefinite: $\gamma=\gamma^\dagger$ and $\gamma\succeq0$.
\end{definition}

The channel at time $t$ is $\cE_t=e^{t\cL}$, and $\rho(t)=e^{t\cL}(\rho(0))$ is the unique solution of the master equation
\begin{align*}
    \odv{\rho(t)}{t}=\cL(\rho(t)).
\end{align*}

We call $\{h_a\}$ the \emph{Hamiltonian coefficients} and $\{\gamma_{ab}\}$ the \emph{dissipative coefficients} of $\cL$. Grouping the two sums in \cref{eq:Lindbladian_pauli}, every Lindbladian splits into a \emph{Hamiltonian part} $\cH$ and a \emph{dissipative part} (or \emph{dissipator}) $\cD$:
\begin{equation}\label{eq:lindblad_split}
\cL=\cH+\cD, \qquad
\cH(\rho)=\sum_{a\neq 0}-\iu h_a [P_a,\rho],
\qquad
\cD(\rho)=\sum_{a,b\neq0}\gamma_{ab}\ab\Big(P_a\rho P_b -\frac{1}{2}\{P_b P_a,\rho\}).
\end{equation}

The Hamiltonian and dissipative supports record the nonzero Pauli coefficients of $\cL$.
\begin{definition}[Support of a Lindbladian]\label{def:lindblad_support}
    Let $\cL$ be a Lindbladian with coefficients $\{h_a\}$ and $\{\gamma_{ab}\}$ as in \cref{eq:Lindbladian_pauli}. The \emph{Hamiltonian support} and the \emph{dissipative support} of $\cL$ are
    \[
    \cS_H \coloneq \cbra{a\neq0:h_a\neq0}, \qquad
    \cS_D \coloneq \cbra{(a,b):a,b\neq0,\ \gamma_{ab}\neq0},
    \]
    and the \emph{support} of $\cL$ is the pair $(\cS_H,\cS_D)$. For a threshold $\eta>0$, the \emph{$\eta$-heavy supports} collect the coordinates of magnitude at least $\eta$,
    \[
    \cS_H^{\geq\eta}\coloneq\cbra{a\neq0:\abs{h_a}\geq\eta},
    \qquad
    \cS_D^{\geq\eta}\coloneq\cbra{(a,b):a,b\neq0,\ \abs{\gamma_{ab}}\geq\eta}.
    \]
\end{definition}

A natural measure of the overall strength of the dynamics is the induced operator norm of the generator in the Heisenberg picture.
\begin{definition}[Dynamical strength]\label{def:dynamical_strength}
The \emph{dynamical strength} of a Lindbladian $\cL$ is $\norm{\cL^\dag}_{\infty\to\infty}$.
\end{definition}

Throughout the paper, we assume the known bound $\norm{\cL^\dag}_{\infty\to\infty}\leq\Lambda$ on the dynamical strength.

\subsection{Bell basis and Choi state}\label{sec:bell_choi}
Let $d\coloneq2^n$ denote the dimension of the $n$-qubit Hilbert space $\mathscr{H}$, and let
\[
\ket{\Phi_0} \coloneq \frac{1}{\sqrt{d}} \sum_{j=0}^{d-1} \ket{j} \ket{j}
\]
be the maximally entangled state on $\mathscr{H}\otimes\mathscr{H}$. For an operator $A\in\cB(\mathscr{H})$, write
\[
\ket{\Phi_A} \coloneq (I \otimes A)\ket{\Phi_0}.
\]
Specifically, for a Pauli label $a$, define
\[
\ket{\Phi_a} \coloneq (I\otimes P_a) \ket{\Phi_0}.
\]
By the orthogonality of Pauli operators in \cref{fact:pauli_inprod}, we have
\[
\braket{\Phi_a}{\Phi_b} = \frac{1}{d} \tr(P_aP_b) = \mathbf{1}[a=b].
\]
Thus $\{\ket{\Phi_a}\}_{a\in\F_2^{2n}}$ is an orthonormal basis of $\mathscr{H}\otimes\mathscr{H}$, called the \emph{Bell basis}. We also use the standard identity
\begin{equation}\label{eq:mes_identity}
\bra{\Phi_0}(A\otimes B)\ket{\Phi_0} = \frac{1}{d}\tr(A^\top B)
\qquad\text{for all } A,B\in\cB(\mathscr{H}),
\end{equation}
which follows by expanding both sides in the computational basis. Every quantum channel is faithfully represented by a bipartite state via the Choi--Jamio{\l}kowski isomorphism~\cite{Jam72,Cho75}.

\begin{definition}[Choi state]\label{def:choi_state}
    The \emph{Choi state} of a quantum channel $\cN\colon\cB(\mathscr{H})\to\cB(\mathscr{H})$ is
    \[
        J(\cN)\coloneq(\mathrm{id}\otimes \cN)\ab\big(\ketbra{\Phi_0}{\Phi_0}).
    \]
\end{definition}

Because $\cN$ is completely positive and trace preserving, $J(\cN)$ is a quantum state on $\mathscr{H}\otimes\mathscr{H}$. Moreover, the map $\cN\mapsto J(\cN)$ is linear and injective, so the Choi state determines the channel completely.

\subsection{Derivative estimation by Chebyshev interpolation}\label{sec:chebyshev_interpolation}
For $q\in\N_{>0}$, let $f$ be a real-valued function that is $(q+1)$ times differentiable on a closed interval $\cI$.
Given $q+1$ distinct nodes $x_0,x_1,\dots,x_q\in\cI$, let $p_q$ be the unique polynomial of degree at most $q$ satisfying $p_q(x_j)=f(x_j)$ for $j=0,1,\dots,q$.

We use Chebyshev--Lobatto interpolation, whose nodes cluster near the endpoints of $[-1,1]$.
\begin{definition}[Chebyshev--Lobatto nodes]\label{def:cheb_nodes}
    For $q\in\N_{>0}$, the $q+1$ Chebyshev--Lobatto nodes are
    \[
        x_j=\cos\ab\Big(\frac{j\pi}{q}),\quad j=0,1,\dots,q.
    \]
\end{definition}

Throughout this paper, we estimate endpoint derivatives of the Lindbladian dynamics at $t=0$ using real-time evolution. For $q\in\N_{>0}$ and $T>0$, map the $q+1$ nodes in \cref{def:cheb_nodes} to $[0,T]$ by
\[
t_j=\frac{T}{2}\ab\Big(1-\cos\ab\Big(\frac{j\pi}{q})),\qquad j=0,1,\dots,q.
\]
Thus $t_0=0$ and $t_q=T$. Given the values of $f$ at these nodes, the endpoint derivative of the degree-$q$ interpolant $p_q$ is
\[
p_q'(0)=\sum_{j=0}^q\ell_j'(0)\,f(t_j),
\qquad
\ell_j(t)=\prod_{\substack{0\leq i\leq q\\i\neq j}}\frac{t-t_i}{t_j-t_i},
\]
where $\ell_j$ is the $j$th Lagrange basis polynomial.
The next lemma bounds the interpolation error in this rule.

\begin{lemma}[Endpoint interpolation error]\label{lem:interpolation_error}
Let $f\colon\bbR_{\geq0}\to\bbR$ be $(q+1)$ times differentiable and satisfy $\abs{f^{(q+1)}(x)}\leq K^{q+1}$ for some $K>0$. Mapping $q+1$ Chebyshev--Lobatto nodes to $[0,T]$ with $T=1/K$ then gives
\[
\abs[\big]{f'(0)-p_q'(0)}\leq\frac{K}{(q+1)!}.
\]
For a complex-valued function $f\colon\bbR_{\geq0}\to\bbC$ satisfying the same derivative bound, applying the argument to its real and imaginary parts gives the bound $\sqrt{2}K/(q+1)!$.
\end{lemma}

\begin{proof}
Set
\[
r(t) \coloneq f(t)-p_q(t),
\qquad \omega(t) \coloneq\prod_{j=0}^q(t-t_j).
\]
Since
\[
\omega'(0)=\prod_{j=1}^q(0-t_j)\neq0,
\]
we may define $c\coloneq r'(0)/ \omega'(0)$ and
$F(t)\coloneq r(t) - c\omega(t)$. Then $F(t_j) = 0$ for
$j=0,\ldots,q$ and $F'(0) = 0$. Counting the zero at $0$ twice,
repeated application of Rolle's theorem~\cite[Section~1.1]{BF11} gives a point
$\xi \in (0,T)$ such that $F^{(q+1)}(\xi)=0$. Since
$p_q^{(q+1)} = 0$ and $\omega^{(q+1)} = (q+1)!$, it follows that
\[
0 = F^{(q+1)}(\xi) = f^{(q+1)}(\xi) - c(q+1)!,
\qquad \text{so}\qquad c = \frac{f^{(q+1)}(\xi)}{(q+1)!}.
\]
Therefore
\[
  f'(0) - p_q'(0) = \frac{f^{(q+1)}(\xi)}{(q+1)!} \prod_{j=1}^q(0-t_j).
\]
Using the hypothesis and $t_j\leq T$,
\[
  \abs{f'(0)-p_q'(0)} \leq \frac{K^{q+1}}{(q+1)!}\, T^q = \frac{K}{(q+1)!}.\qedhere
\]
\end{proof}

\begin{lemma}[Endpoint weight bound]\label{lem:weight_bound}
For the Chebyshev--Lobatto nodes on $[0,T]$ with $T=1/K$, the endpoint weights satisfy
\begin{equation}\label{eq:lagrange_polynomial}
    W \coloneq \sum_{j=1}^q \abs{\ell'_j(0)} = \frac{1}{T}\ab\Big(\frac{4q^2 - 1}{3})=O(q^2K).
\end{equation}
\end{lemma}
\begin{proof}
For the Chebyshev--Lobatto nodes mapped to $[0,T]$, the endpoint weights satisfy
\[
\abs{\ell'_j(0)} = \frac{2}{T} \csc^2 \ab\Big(\frac{j\pi}{2q}) \quad (1\leq j<q),\qquad \abs{\ell'_q(0)} = \frac{1}{T}.
\]
This formula is given in~\cite[Section~3]{CHQZ06}.
Using
\[
\sum_{j=1}^{q-1}\csc^2\ab\Big(\frac{j\pi}{2q}) = \frac{2(q^2-1)}{3},
\]
we obtain
\[
W = \frac{2}{T} \ab(\frac{2(q^2-1)}{3}) + \frac{1}{T} = \frac{4q^2-1}{3T}.
\]
Substituting $T=1/K$ gives $W=O(q^2K)$.
\end{proof}

\subsection{Concentration inequalities}\label{sec:concentration}
Statistical errors in this paper are controlled by the following standard concentration bounds.

\begin{fact}[Hoeffding's inequality~\cite{Hoe63}]\label{fact:hoeffding}
    Let $X_1,\dots,X_N$ be independent real random variables with $X_i\in[\alpha_i,\beta_i]$ almost surely, and let $S=\sum_{i=1}^NX_i$. Then for every $s>0$,
    \begin{equation}
        \Pr\sbra[\Big]{\abs[\big]{S-\E[S]}\geq s}\leq2\exp\rbra[\Bigg]{-\frac{2s^2}{\sum_{i=1}^N(\beta_i-\alpha_i)^2}}.
    \end{equation}
\end{fact}

\begin{fact}[Median-of-means bound; see, e.g.,~\cite{JVV86}]\label{fact:mom}
    Let $X_1,\dots,X_N$ be i.i.d.\ real random variables with mean $\mu$ and $\E[X_1^2]\leq\sigma^2$. Partition $[N]$ into $K\coloneq\lceil8\log(1/\delta)\rceil$ batches of equal size, and let $\hat\mu$ be the median of the $K$ batch means. If $N\geq K\,\lceil4\sigma^2/\eta^2\rceil$, then
    \[
    \Pr\ab[\abs{\hat\mu-\mu}>\eta]\leq\delta.
    \]
    In particular, $N=O\ab\big(\sigma^2\eta^{-2}\log(1/\delta))$ samples suffice to achieve accuracy $\eta$ with failure probability $\delta$.
\end{fact}

\section{Threshold support learning for Lindbladians}\label{sec:support_learning}
We construct candidate supports containing every Hamiltonian and dissipative coordinate of magnitude at least $\eta$.
The Bell-coherence identities in \cref{sec:bell_coherence} and the interpolation and positivity bounds in \cref{sec:endpoint_to_heavy_labels} reduce this task to finding heavy labels in a finite-time Pauli error distribution.
The Pauli-twirled experiments in \cref{sec:heavy_label_marginals} sample the two label marginals, whose Cartesian product supplies the candidate labels.

\subsection{Bell coherences and Lindbladian coefficients}\label{sec:bell_coherence}

We use an alternative linear representation of quantum dynamics. Every linear map on $n$-qubit operators admits a unique expansion in the left--right Pauli basis, called the $\chi$-matrix (or process-matrix) representation. The time-evolution channel $e^{t\cL}$ can therefore be written as
\[
 e^{t\cL}(\rho) = \sum_{a,b} \chi_{ab}(t) P_a \rho P_b.
\]
Similarly, we can represent the Lindbladian generator $\cL$ as
\begin{equation}\label{eq:chi}
  \cL(\rho)=\sum_{a,b}\chi_{ab}\,P_a \rho P_b.
\end{equation}
The physical coefficients of $\cL$ can be read off from its $\chi$-matrix as follows.

\begin{lemma}[Coefficients from the $\chi$-matrix]\label{lem:chi_to_coeff}
For every nonidentity Pauli label $a$,
\[
  h_a = \frac{\iu}{2}\ab\big(\chi_{a0}-\chi_{0a}).
\]
For every pair of nonidentity Pauli labels $a,b$,
\[
  \gamma_{ab}=\chi_{ab}.
\]
\end{lemma}
\begin{proof}
Expand each term of the Lindbladian in the left--right Pauli basis of \cref{eq:chi}. The two-sided jump term $\gamma_{ab}\,P_a\rho P_b$ contributes $\gamma_{ab}$ to the entry $\chi_{ab}$. Since the dissipative sum ranges over $a,b\neq0$, these are the only contributions to entries with both indices nonzero, and hence $\chi_{ab}=\gamma_{ab}$ for every pair of nonidentity labels $a,b$. The remaining terms contribute only to entries with a zero index. The Hamiltonian term $-\iu h_c[P_c,\rho]=-\iu h_c(P_c\rho P_0-P_0\rho P_c)$ contributes $-\iu h_c$ to $\chi_{c0}$ and $+\iu h_c$ to $\chi_{0c}$. The anticommutator term, using $P_bP_a=\xi_{ba}P_{a\oplus b}$ from \cref{fact:xi_ab}, is
\[
-\frac{\gamma_{ab}}{2} \{P_bP_a,\rho\}
=-\frac{\xi_{ba}\gamma_{ab}}{2} \ab\big(P_{a\oplus b} \rho P_0+P_0\rho P_{a\oplus b}),
\]
which contributes the same amount $-\tfrac{1}{2}\xi_{ba}\gamma_{ab}$ to both $\chi_{a\oplus b,0}$ and $\chi_{0,a\oplus b}$. After summing all contributions, we obtain, for every $c\neq0$,
\[
\chi_{c0}=-\iu h_c-\frac{1}{2}\sum_{\substack{(a,b)\in\cS_D\\ a\oplus b=c}}\xi_{ba}\gamma_{ab},
\qquad
\chi_{0c}=\iu h_c-\frac{1}{2}\sum_{\substack{(a,b)\in\cS_D\\ a\oplus b=c}}\xi_{ba}\gamma_{ab},
\]
so the anticommutator sums cancel in the difference, and $\frac{\iu}{2}(\chi_{c0}-\chi_{0c})=\frac{\iu}{2}(-2\iu h_c)=h_c$.
\end{proof}

The coefficient identities also bound the number of heavy coordinates using only the dynamical strength.

\begin{lemma}[Effective sparsity]\label{lem:effective_sparsity}
Let $\cL$ be an $n$-qubit Lindbladian with $\norm{\cL^\dag}_{\infty\to\infty}\leq\Lambda$ and $\chi$-matrix as in~\cref{eq:chi}.
Then
\begin{align}\label{eq:coefficient_energy}
2\sum_{a\neq0}\abs{h_a}^2+\sum_{a,b\neq0}\abs{\gamma_{ab}}^2&\leq\sum_{a,b}\abs{\chi_{ab}}^2\leq\Lambda^2.
\end{align}
Consequently, for every $\eta>0$,
\begin{align}\label{eq:effective_sparsity}
2\abs{\cS_H^{\geq\eta}}+\abs{\cS_D^{\geq\eta}}\leq\frac{\Lambda^2}{\eta^2}.
\end{align}
\end{lemma}
\begin{proof}
Let $d=2^n$ and let $\mathsf L=\sum_{a,b}\chi_{ab}(P_b^\top\otimes P_a)$ be the matrix representing $\cL$ under column vectorization.
The matrices $P_b^\top\otimes P_a$ are Hilbert--Schmidt orthogonal with squared norm $d^2$, and $\{P_c/\sqrt d\}_c$ is an orthonormal basis of the operator space.
Thus
\begin{align}
d^2\sum_{a,b}\abs{\chi_{ab}}^2&=\norm{\mathsf L}_2^2=\norm{\mathsf L^\dag}_2^2=\frac{1}{d}\sum_c\norm{\cL^\dag(P_c)}_2^2\nonumber\\
&\leq\sum_c\norm{\cL^\dag(P_c)}_\infty^2\leq d^2\Lambda^2,
\end{align}
where the last inequality uses $\norm{P_c}_\infty=1$.
By~\cref{lem:chi_to_coeff}, $\gamma_{ab}=\chi_{ab}$ for $a,b\neq0$ and $2\abs{h_a}^2=\frac12\abs{\chi_{a0}-\chi_{0a}}^2\leq\abs{\chi_{a0}}^2+\abs{\chi_{0a}}^2$.
Summing these inequalities proves~\cref{eq:coefficient_energy}.
Each $\eta$-heavy coefficient contributes at least $\eta^2$ to its corresponding sum, giving~\cref{eq:effective_sparsity}.
\end{proof}

For $t\geq0$, we apply the Choi--Jamio{\l}kowski isomorphism of \cref{def:choi_state} to the dynamical semigroup $\{e^{t\cL}\}_{t\geq0}$ and define its Choi state by
\[
J_t\coloneq J\ab\big(e^{t\cL})=(\mathrm{id}\otimes e^{t\cL})\ab\big(\ketbra{\Phi_0}{\Phi_0}).
\]
For Pauli labels $a,b$, define the \emph{Bell coherences} and the \emph{Pauli error rates} of the evolution as
\[
g_{ab}(t) \coloneq \bra{\Phi_a} \,J_t\,\ket{\Phi_b},
\qquad \lambda_a(t) \coloneq g_{aa}(t).
\]
Since $J_t$ is a quantum state and the Bell basis is orthonormal, the Pauli error rates $\{\lambda_a(t)\}_{a\in \F_2^{2n}}$ form a probability distribution. At $t=0$, we have $J_0 = \ketbra{\Phi_0}{\Phi_0}$, and the endpoint derivative is
\[
J'_0 \coloneq \odv[delims-eval=.|]{J_t}{t}_{t=0}=(\mathrm{id} \otimes\cL)(\ketbra{\Phi_0}{\Phi_0}).
\]
In the Bell basis, the endpoint derivative is exactly the $\chi$-matrix of $\cL$.

\begin{lemma}[Bell matrices of the dynamics]\label{lem:bell_chi}
For all Pauli labels $a,b\in\F_2^{2n}$ and all $t\geq0$,
\[
g_{ab}(t)=\bra{\Phi_a}\,J_t\,\ket{\Phi_b}=\chi_{ab}(t),
\qquad
g'_{ab}(0)=\bra{\Phi_a}\,J'_0\,\ket{\Phi_b}=\chi_{ab}.
\]
Consequently, for all nonidentity Pauli labels $a,b$,
\[
h_a = \Imag\bra{\Phi_0}\, J'_0 \, \ket{\Phi_a} = \Imag g'_{0a}(0), \qquad \gamma_{ab}=\bra{\Phi_a}\, J'_0\, \ket{\Phi_b} = g'_{ab}(0).
\]
\end{lemma}
\begin{proof}
Expanding $e^{t\cL}$ in the left--right Pauli basis and using the Hermiticity of every Pauli operator, we obtain
\[
J_t = (\mathrm{id} \otimes e^{t\cL})(\ketbra{\Phi_0}{\Phi_0})
= \sum_{a,b}\chi_{ab}(t) \,(I\otimes P_a) \ketbra{\Phi_0}{\Phi_0}(I\otimes P_b)
=\sum_{a,b} \chi_{ab}(t) \ketbra{\Phi_a}{\Phi_b}.
\]
The orthonormality of the Bell basis gives $g_{ab}(t)=\chi_{ab}(t)$. Applying the same calculation to the generator expansion yields
\[
J'_0 = (\mathrm{id} \otimes \cL)(\ketbra{\Phi_0}{\Phi_0})
= \sum_{a,b}\chi_{ab} \,(I\otimes P_a) \ketbra{\Phi_0}{\Phi_0}(I\otimes P_b)
=\sum_{a,b} \chi_{ab}\ketbra{\Phi_a}{\Phi_b}.
\]
Since $J_t$ is Hermitian for every $t$, its derivative $J'_0$ is Hermitian, so
\[
\chi_{c0} = \overline{\chi_{0c}}, \qquad h_c = \frac{\iu}{2}(\chi_{c0} - \chi_{0c}) = \frac{\iu}{2} \ab\big(-2\iu\Imag\chi_{0c})=\Imag\chi_{0c}. \qedhere
\]
\end{proof}
\subsection{From endpoint derivatives to heavy labels}\label{sec:endpoint_to_heavy_labels}

Support learning uses endpoint-derivative interpolation to relate large coefficients to finite-time Bell populations.
Coefficient learning uses it to estimate the coefficients directly.
The following uniform derivative bound controls both uses.

\begin{lemma}[Derivative bound for Choi matrix elements]\label{lem:choi_derivative_bound}
Let $\cL$ be a Lindbladian with dynamical strength at most $\Lambda$. For $A,B\in\cB(\mathscr{H})$, define the signal $f_{A,B}(t)\coloneq\bra{\Phi_A}\,J_t\,\ket{\Phi_B}$. Then for every $r \in \N$ and $t \geq 0$,
\[
\abs[\big]{f_{A,B}^{(r)}(t)} \leq \norm{A}_\infty\norm{B}_\infty\,\Lambda^r.
\]
\end{lemma}

\begin{proof}
Differentiating the semigroup gives
\[
\odv[r]{J_t}{t} = (\mathrm{id} \otimes \cL^r e^{t\cL})(\ketbra{\Phi_0}{\Phi_0}).
\]
Expanding $\ketbra{\Phi_0}{\Phi_0}$ in the computational basis as
\[
\ketbra{\Phi_0}{\Phi_0} = \frac{1}{d} \sum^{d-1}_{j,l=0} \ketbra{j}{l} \otimes \ketbra{j}{l}
\]
and applying the identity in \cref{eq:mes_identity} to each term,
\[
f_{A,B}^{(r)}(t) = \frac{1}{d} \sum^{d-1}_{j,l=0} \bra{\Phi_0}\left(\ketbra{j}{l} \otimes \ab\big(A^\dagger \cL^r e^{t\cL}\ab(\ketbra{j}{l} ) B)\right)\ket{\Phi_0}
=\frac{1}{d^2} \sum_{j,l=0}^{d-1} \bra{j}A^\dagger \cL^r e^{t\cL}\ab(\ketbra{j}{l}) B\ket{l}.
\]
To justify trace-norm contractivity on the generally non-Hermitian matrix unit $\ketbra{j}{l}$, write $\cE_t\coloneq e^{t\cL}$. Its adjoint $\cE_t^\dagger$ is unital and completely positive. The Kadison--Schwarz inequality gives, for every operator $X$,
\[
  \cE_t^\dagger(X)^\dagger \cE_t^\dagger(X) \preceq \cE_t^\dagger(X^\dagger X) \preceq \norm{X}_\infty^2 I,
\]
where the second inequality uses positivity and unitality. On the other hand, unitality gives $\cE_t^\dagger(I) = I$. Hence $\norm{\cE_t^\dagger}_{\infty\to\infty}=1$, and duality of induced Schatten norms implies
\[
  \norm{\cE_t}_{1\to1}=1
\]
on the full operator space, not only on Hermitian inputs. Therefore each summand is bounded by
\begin{align*}
\abs[\big]{\bra{j} A^\dagger \cL^r \cE_t\ab(\ketbra{j}{l}) B\ket{l}}
&\leq \norm{A}_\infty \norm{B}_\infty
      \norm[\big]{\cL^r\cE_t\ab(\ketbra{j}{l})}_\infty\\
&\leq \norm{A}_\infty \norm{B}_\infty
      \norm{\cL}_{1\to1}^{r}
      \norm[\big]{\cE_t\ab(\ketbra{j}{l})}_1\\
&\leq \norm{A}_\infty \norm{B}_\infty
      \Lambda^r\norm{\ketbra{j}{l}}_1\\
&=\norm{A}_\infty \norm{B}_\infty\Lambda^r,
\end{align*}
where we also used $\norm{Y}_\infty\leq\norm{Y}_1$, $\norm{\cL}_{1\to1}=\norm{\cL^\dagger}_{\infty\to\infty}\leq\Lambda$, and $\norm{\ketbra{j}{l}}_1=1$. Averaging the $d^2$ terms with weight $1/d^2$ proves the claim.
\end{proof}

By \cref{lem:choi_derivative_bound}, the Bell coherence $g_{ab}(t)$ for Pauli labels $a$ and $b$ satisfies
\begin{equation}\label{eq:coherence_derivative}
\abs[\big]{g^{(r)}_{ab}(t)} \leq \norm{P_a}_\infty\norm{P_b}_\infty\,\Lambda^r = \Lambda^r.
\end{equation}

Fix a threshold $\eta \in (0, \Lambda]$. We choose the Chebyshev interpolation parameters as
\[
T = \frac{1}{\Lambda}, \qquad
q = \min\cbra[\big]{q'\in\N_{>0}:(q'+1)!\geq 4\Lambda/\eta}.
\]
For $1 \leq j \leq q$, let $t_j$ be the Chebyshev--Lobatto nodes and let $\ell'_j(0)$ be the associated endpoint weights. For ease of notation, define
\[
W \coloneq \sum_{j=1}^q \abs{\ell'_j(0)},
\qquad
p_j \coloneq \frac{\abs{\ell'_j(0)}}{W},
\qquad
\alpha \coloneq \frac{\eta^2}{4W^2}.
\]
Define the \emph{Chebyshev mixture} by
\begin{equation}\label{eq:mixture}
 \mu(a) \coloneq \sum_{j=1}^q p_j \,\lambda_a(t_j),
 \qquad \text{for } a\in\F_2^{2n}.
\end{equation}
Since $\{p_j\}_{j=1}^q$ is a probability distribution and each $\{\lambda_a(t_j)\}_a$ is a probability distribution, $\{\mu(a)\}_{a\in \F_2^{2n}}$ is also a probability distribution over Pauli labels. The next lemma shows that heavy coefficients lead to heavy entries in $\mu$.

\begin{lemma}[Heavy coherences imply heavy populations]\label{lem:coh_to_pop}
For every nonidentity Pauli label $a$, if $\abs{h_a}\geq\eta$, then $\mu(a)\geq\alpha$. For every pair of nonidentity Pauli labels $a,b$, if $\abs{\gamma_{ab}}\geq\eta$, then $\mu(a)\geq\alpha$ and $\mu(b)\geq\alpha$.
\end{lemma}
\begin{proof}
Fix a pair of labels $(a,b)\neq(0,0)$. By \cref{lem:choi_derivative_bound} and our choices of $T$ and $q$, the interpolation-error bound of \cref{lem:interpolation_error} gives
\[
\abs[\Big]{g'_{ab}(0)-\sum_{j=1}^q \ell'_j(0) \,g_{ab}(t_j)} \leq \frac{\sqrt{2} \Lambda}{(q+1)!} \leq\frac{\eta}{2}.
\]
Because $g_{ab}(0) = 0$, the $t_0 = 0$ term vanishes. Suppose that $\abs{g'_{ab}(0)}\geq\eta$. The triangle inequality followed by the Cauchy--Schwarz inequality for the probability weights $p_j$ gives
\begin{equation}\label{eq:weighted_coherence}
\frac{\eta}{2} \leq \abs[\Big]{\sum_{j=1}^q \ell'_j(0)\,g_{ab}(t_j)}
\leq W \sum_{j=1}^q p_j\, \abs{g_{ab}(t_j)}
\leq W \sqrt{\sum_{j=1}^q p_j\, \abs{g_{ab}(t_j)}^2}.
\end{equation}
Since $J_t\succeq0$ has unit trace, if $a\neq b$, the $2\times2$ principal submatrix of $J_t$ on the orthonormal Bell vectors $\ket{\Phi_a}$ and $\ket{\Phi_b}$ is positive semidefinite, and every diagonal entry is at most one. If $a=b$, then $g_{aa}(t)=\lambda_a(t)$ and $0\leq\lambda_a(t)\leq1$. Thus
\begin{equation}\label{eq:coherence_to_error_rates}
\abs{g_{ab}(t)}^2 \leq \lambda_a(t)\lambda_b(t) \leq \min\cbra{\lambda_a(t), \lambda_b(t)}.
\end{equation}
Combining \cref{eq:weighted_coherence,eq:coherence_to_error_rates} gives
\[
\mu(a) =  \sum_{j=1}^q p_j\,\lambda_a(t_j) \geq  \sum_{j=1}^q p_j\,\abs{g_{ab}(t_j)}^2 \geq \frac{\eta^2}{4 W^2} = \alpha,
\]
and the same holds for $\mu(b)$.

For a dissipative pair with $\abs{\gamma_{ab}}\geq\eta$, \cref{lem:bell_chi,lem:chi_to_coeff} give $g'_{ab}(0)=\chi_{ab}=\gamma_{ab}$, so the preceding argument applies to $(a,b)$ and yields $\mu(a),\mu(b)\geq\alpha$.
When $a=b$, this conclusion reduces to $\mu(a)\geq\alpha$.
For a Hamiltonian label with $\abs{h_a}\geq\eta$, we have $\abs{g'_{0a}(0)}\geq\abs{\Imag\chi_{0a}}=\abs{h_a}\geq\eta$, and the preceding argument applied to $(0,a)$ yields $\mu(a)\geq\alpha$.
\end{proof}

\subsection{Sampling heavy-label marginals}\label{sec:heavy_label_marginals}
The next step is to sample the $x$- and $z$-marginals of $\mu$, which suffice to construct a candidate set containing every heavy label. These marginals can be sampled by implementing the Pauli channel whose error distribution is $\mu$.
\begin{lemma}[Twirled channel realizing the mixture]\label{lem:twirled_channel_mixture}
    Define the following channel by twirling and mixing the semigroup $\{e^{t\cL}\}_{t\geq0}$:
    \[
    \closure{\cE}(\rho) \coloneq \sum_{j=1}^q \frac{p_j}{4^n} \sum_{c\in \F_2^{2n}} P_c e^{t_j \cL}\ab(P_c \rho P_c) P_c.
    \]
    Then
    \[
    \closure{\cE}(\rho) = \sum_{a \in \F_2^{2n}} \mu(a) P_a \rho P_a.
    \]
\end{lemma}
\begin{proof}
Write $e^{t_j \cL}(\rho) = \sum_{a,b}\chi_{ab}(t_j) P_a \rho P_b$. Applying \cref{lem:bell_chi} to the finite-time evolution $e^{t_j \cL}$ gives $\chi_{ab}(t_j) = g_{ab}(t_j)$ and $\chi_{aa}(t_j) = \lambda_{a}(t_j)$. Conjugation by a Pauli operator $P_c$ yields
\begin{align*}
P_c e^{t_j \cL}\ab(P_c \rho P_c) P_c &= \sum_{a,b}\chi_{ab}(t_j) P_c P_a P_c \rho P_c P_b P_c \\
&= \sum_{a,b}\chi_{ab}(t_j) \ab\big((-1)^{\sinprod{a}{c}}P_a) \rho \ab\big((-1)^{\sinprod{b}{c}} P_b)\\
&= \sum_{a,b} (-1)^{\sinprod{a\oplus b}{c}} \chi_{ab}(t_j) P_a \rho P_b.
\end{align*}
By the symplectic character orthogonality in \cref{fact:char_orth}, averaging over the uniformly random twirling label $c \in \F_2^{2n}$ gives
\begin{align*}
\frac{1}{4^n} \sum_{c \in \F_2^{2n}} P_c e^{t_j \cL}\ab(P_c \rho P_c) P_c &= \sum_{a,b} \chi_{ab}(t_j) \ab\Big(\frac{1}{4^n} \sum_{c\in \F_2^{2n}} (-1)^{\sinprod{a\oplus b}{c}}) P_a \rho P_b\\
&= \sum_{a,b} \chi_{ab}(t_j) \ab\big(\mathbf{1}[a\oplus b = 0]) P_a \rho P_b\\
&= \sum_{a \in \F_2^{2n}} \lambda_{a}(t_j) P_a \rho P_a.
\end{align*}
Averaging these twirled channels over $j$ with weights $p_j$ gives
\[
\closure{\cE}(\rho)
=\sum_{j=1}^q p_j \sum_{a \in \F_2^{2n}} \lambda_a(t_j) P_a\rho P_a
=\sum_{a}\ab\Big(\sum_{j=1}^q p_j \lambda_a(t_j)) P_a\rho P_a
=\sum_{a}\mu(a)P_a\rho P_a. \qedhere
\]
\end{proof}

\begin{algorithm}[ht!]
\caption{Support learning by displacement sampling}\label{alg:displacement}
\begin{algorithmic}[1]
\Require Access to $\{e^{t\cL}\}_{t\geq0}$; threshold $\eta$; failure probability $\delta$; dynamical strength bound $\Lambda$.
\Ensure Candidate supports $\hat\cS_H$ and $\hat\cS_D$.
\State Set $T\gets1/\Lambda$ and $q\gets\min\cbra{q'\in\N_{>0}:(q'+1)!\geq4\Lambda/\eta}$; compute the nodes $t_1,\ldots,t_q$ and weights $\ell'_1(0),\ldots,\ell'_q(0)$.
\State Set $W \gets \sum_{j=1}^q \abs{\ell'_j(0)}$, $p_j \gets  \abs{\ell'_j(0)}/W$, and $\alpha \gets \eta^2/(4 W^2)$.
\State Set $N_{\mathrm{supp}}\gets\ceil[\big]{\alpha^{-1}\log(2/(\alpha\delta))}$ and $\hat\cX\gets\emptyset$, $\hat\cZ\gets\emptyset$.
\For{$m=1,\ldots,N_{\mathrm{supp}}$}
  \State Sample $j\in[q]$ with probability $p_j$, a uniformly random label $c\in\F_2^{2n}$, and a uniformly random bit string $u\in\{0,1\}^n$.
  \State Prepare $\ket{u}$, apply $P_c$, evolve under $e^{t_j\cL}$, and apply $P_c$ again.
  \State Measure each qubit in the computational basis, obtaining outcome $v$.
  \State Add $u\oplus v$ to $\hat\cX$.
\EndFor
\For{$m=1,\ldots,N_{\mathrm{supp}}$}
  \State Sample $j\in[q]$ with probability $p_j$, a uniformly random label $c\in\F_2^{2n}$, and a uniformly random bit string $u\in\{0,1\}^n$.
  \State Prepare $\ket{u}_X$, apply $P_c$, evolve under $e^{t_j\cL}$, and apply $P_c$ again.
  \State Measure each qubit in the $X$ basis, obtaining outcome $v$.
  \State Add $u\oplus v$ to $\hat\cZ$.
\EndFor
\State Set $\hat\cS\gets\cbra{(z,x)\in\F_2^{2n}:z\in\hat\cZ,\ x\in\hat\cX}\setminus\{0\}$.
\State \Return $\hat\cS_H \gets \hat\cS$ and $\hat\cS_D \gets \hat\cS \times \hat\cS$.
\end{algorithmic}
\end{algorithm}

A Pauli operator produces observable displacements in the computational and $X$ bases, allowing us to sample the two marginals of $\mu$.

\begin{lemma}[Displacement identities]\label{lem:displacement_identity}
Write a Pauli label as $a=(z,x)\in\F_2^{2n}$, and for $u\in\{0,1\}^n$ let $\ket{u}_X\coloneq H^{\otimes n}\ket{u}$ denote the $X$-basis states, where $H$ is the Hadamard gate. Then
\begin{equation*}
P_a\ket{u} = \omega_{a,u}\ket{u\oplus x},
\qquad
P_a\ket{u}_X = \widetilde\omega_{a,u}\ket{u\oplus z}_X
\end{equation*}
for phases $\omega_{a,u},\widetilde\omega_{a,u}\in\bbC$ satisfying $\abs{\omega_{a,u}}=\abs{\widetilde\omega_{a,u}}=1$.
\end{lemma}
\begin{proof}
By the binary representation in \cref{def:binary_rep},
\[
P_a=\vartheta(z,x)Z^zX^x,\qquad \vartheta(z,x)\coloneq\prod_{i=1}^n(-\iu)^{z_ix_i}.
\]
Here, $Z^{z}\coloneq\bigotimes_{i=1}^nZ^{z_i}$ and $X^{x}\coloneq\bigotimes_{i=1}^nX^{x_i}$. On the computational basis, $X^{x}\ket{u}=\ket{u\oplus x}$ and $Z^{z}\ket{v}=(-1)^{z\cdot v}\ket{v}$ for every $v\in\{0,1\}^n$, so
\[
P_a\ket{u} = \vartheta(z,x)(-1)^{z\cdot(u\oplus x)}\ket{u\oplus x},
\]
and hence $\omega_{a,u}=\vartheta(z,x)(-1)^{z\cdot(u\oplus x)}$.
For the $X$ basis, the conjugation rules $HXH=Z$ and $HZH=X$ give $X^{x}H^{\otimes n} = H^{\otimes n}Z^{x}$ and $Z^{z}H^{\otimes n} = H^{\otimes n}X^{z}$. Hence
\begin{align*}
P_a\ket{u}_X
&=\vartheta(z,x)Z^{z}X^{x}H^{\otimes n}\ket{u}\\
&=\vartheta(z,x)(-1)^{x\cdot u}Z^{z}H^{\otimes n}\ket{u}\\
&=\vartheta(z,x)(-1)^{x\cdot u}H^{\otimes n}\ket{u\oplus z}\\
&=\vartheta(z,x)(-1)^{x\cdot u}\ket{u\oplus z}_X.
\end{align*}
Thus $\widetilde\omega_{a,u}=\vartheta(z,x)(-1)^{x\cdot u}$.
\end{proof}

By \cref{lem:twirled_channel_mixture,lem:displacement_identity}, the displacement recorded in a computational-basis experiment is distributed according to the $x$-marginal of $\mu$, while an $X$-basis experiment samples its $z$-marginal.
For every label $a=(z,x)$ with $\mu(a)\geq\alpha$, both marginal probabilities are at least $\alpha$.
Once the sampled sets contain all such marginal components, their Cartesian product contains every $\alpha$-heavy label.

\begin{theorem}[Threshold support learning]\label{thm:support_learning}
Let $\cL$ be an arbitrary $n$-qubit Lindbladian satisfying $\norm{\cL^\dag}_{\infty\to\infty}\leq\Lambda$, and let $\cS_H^{\geq\eta}$ and $\cS_D^{\geq\eta}$ be the $\eta$-heavy supports defined in \cref{def:lindblad_support}. For every $\eta\in(0,\Lambda]$ and $\delta\in(0,1)$, \cref{alg:displacement} outputs $\hat\cS_H$ and $\hat\cS_D$ such that, with probability at least $1-\delta$,
\[
\cS_{H}^{\geq \eta} \subseteq \hat\cS_{H}, \qquad \cS_{D}^{\geq \eta} \subseteq \hat\cS_{D}.
\]
Moreover, the output size deterministically satisfies $\abs{\hat\cS_H} + \abs{\hat\cS_D} = \widetilde O\ab\big(\Lambda^{8} \eta^{-8}\log^4(1/\delta))$. The number of experiments is
\[
\widetilde O\ab\Big(\frac{\Lambda^2}{\eta^2}\log\frac{1}{\delta}).
\]
The total evolution time is $\Ttot=\widetilde O\ab\big(\Lambda\eta^{-2}\log(1/\delta))$. The classical running time is $\widetilde O\ab\big(n\Lambda^{8}\eta^{-8}\log^4(1/\delta))$.
\end{theorem}

\begin{proof}
Because measurement probabilities are linear in the channel, averaging a single experiment over the random node $j$ and twirling label $c$ yields exactly the displacement distribution of the mixture channel $\closure{\cE}$ in \cref{lem:twirled_channel_mixture}.
By the displacement identities in \cref{lem:displacement_identity}, a computational-basis round records $u\oplus v=x$ with probability $\mu_X(x)\coloneq\sum_{z\in\F_2^{n}}\mu(z,x)$, and an $X$-basis round records $u\oplus v=z$ with probability $\mu_Z(z)\coloneq\sum_{x\in\F_2^{n}}\mu(z,x)$.

Since $\mu_X$ and $\mu_Z$ are probability distributions, each has at most $1/\alpha$ labels of mass at least $\alpha$. A fixed such label is missed by all $N_{\mathrm{supp}}=\ceil{\alpha^{-1}\log(2/(\alpha\delta))}$ shots in the corresponding batch with probability at most
\[
(1-\alpha)^{N_{\mathrm{supp}}} \leq e^{-\alpha N_{\mathrm{supp}}} \leq \frac{\alpha\delta}{2}.
\]
A union bound over the at most $2/\alpha$ heavy labels across the two marginal distributions shows that, with probability at least $1-\delta$, every $x$ with $\mu_X(x)\geq\alpha$ lies in $\hat\cX$ and every $z$ with $\mu_Z(z)\geq\alpha$ lies in $\hat\cZ$. On this event, every label $r=(z,x)$ with $\mu(r)\geq\alpha$ satisfies $\mu_X(x)\geq\mu(r)\geq\alpha$ and $\mu_Z(z)\geq\alpha$, hence $r\in\hat\cS$.

By \cref{lem:coh_to_pop}, every Hamiltonian label $a$ with $\abs{h_a}\geq\eta$ satisfies $\mu(a)\geq\alpha$ and hence lies in $\hat\cS$. Similarly, if $\abs{\gamma_{ab}}\geq\eta$, then both $a$ and $b$ satisfy $\mu(a),\mu(b)\geq\alpha$ and hence lie in $\hat\cS$.

By \cref{lem:weight_bound}, $W=O(q^2\Lambda)$ with $q=\widetilde\Theta(1)$, and $\alpha^{-1}=4W^2/\eta^2$, so the number of experiments is
\[
2N_{\mathrm{supp}} = 2\ceil[\Big]{\frac{1}{\alpha}\log\frac{2}{\alpha\delta}} = \widetilde O\ab\Big(\frac{\Lambda^2}{\eta^{2}} \log\frac{1}{\delta}).
\]
Every experiment evolves for time at most $T=1/\Lambda$, giving $\Ttot\leq 2N_{\mathrm{supp}}/\Lambda$. Because each batch records at most $N_{\mathrm{supp}}$ distinct displacements, $\abs{\hat\cZ},\abs{\hat\cX}\leq N_{\mathrm{supp}}$. Consequently, $\abs{\hat\cS}=O(N_{\mathrm{supp}}^2)$, and
\[
\abs{\hat\cS_H} + \abs{\hat\cS_D} = O(N_{\mathrm{supp}}^4)= \widetilde O\ab\Big(\frac{\Lambda^{8}}{\eta^{8}}\log^4\ab\big(\frac{1}{\delta})).
\]
Recording the displacements costs $O(nN_{\mathrm{supp}})$, and constructing the candidate sets costs $O\ab\big(n(\abs{\hat\cS_D}+\abs{\hat\cS_H}))=\widetilde O\ab\big(n\Lambda^{8}\eta^{-8}\log^4(1/\delta))$.
\end{proof}

\section{Learning Lindbladian coefficients by Clifford probing}\label{sec:shadow}

Given candidate supports, we express their coefficients as endpoint derivatives of Bell-basis matrix elements~(\cref{sec:choi_witness}).
We estimate these derivatives from a shared collection of classical-shadow snapshots~(\cref{sec:process_shadow}).
Each experiment uses Clifford preparation and measurement circuits on the original system alone, with one uninterrupted application of the unknown evolution between them.
The Choi state is used only to define the signals.
The complete algorithm and its guarantee are given in~\cref{sec:shadow_alg}.

\subsection{The Choi state and coefficient witnesses}\label{sec:choi_witness}

Recall from \cref{lem:bell_chi} that every coefficient of $\cL$ can be expressed as a linear combination of endpoint derivatives of Bell-basis matrix elements of the Choi state $J_t$. Real and imaginary parts of off-diagonal matrix elements are obtained from diagonal expectations by polarization, so all observables remain rank-one projectors.

We call an operator $A\in\cB(\mathscr{H})$ \emph{normalized} if $\tfrac{1}{d}\tr(A^\dagger A)=1$. Then
\[
\ket{\Phi_A} \coloneq (I\otimes A)\ket{\Phi_0}
\]
is a normalized state. Define the projector and its finite-time signal
\[
\Pi_A \coloneq \ketbra{\Phi_A}{\Phi_A},
\qquad
f_A(t) \coloneq \bra{\Phi_A}\, J_t \,\ket{\Phi_A} = \tr\ab\big(\Pi_A J_t).
\]
We call the operators collected below \emph{witnesses} because their signals are used to recover specific coefficients of $\cL$.

\begin{definition}[Witness family]\label{def:witness}
The \emph{witness family} $\cW$ consists of the following normalized operators:
\begin{itemize}
    \item the \emph{diagonal witnesses}: $P_a$ for every label $a\in\F_2^{2n}$, including $P_0=I$;
    \item the \emph{pair witnesses}: for every ordered pair $(a,b)$ of distinct labels $a \neq b$,
    \[
    A_{ab} \coloneq \frac{P_a+P_b}{\sqrt{2}},
    \qquad
    A^{\iu}_{ab} \coloneq \frac{P_a+\iu P_b}{\sqrt{2}}.
    \]
\end{itemize}
\end{definition}

Every $A\in\cW$ is normalized and satisfies $1\leq\norm{A}_\infty\leq\sqrt{2}$, and $\abs{\cW}\leq2\cdot16^n$.
Normalization holds because distinct Pauli operators are orthogonal.
For example,
\[
\tr\ab\big((A^{\iu}_{ab})^\dagger A^{\iu}_{ab})=\frac{1}{2}\tr(2I+\iu[P_a,P_b])=d.
\]
The operator-norm bounds follow from the triangle inequality and $\norm{A}_\infty^2\geq\frac{1}{d}\tr(A^\dagger A)=1$.

\begin{lemma}[Witness signals]\label{lem:witness_signals}
For all distinct labels $a \neq b$ and all $t\geq0$,
\[
f_{A_{ab}}(t)=\frac{f_{P_a}(t)+f_{P_b}(t)}{2}+\Real\bra{\Phi_a}\, J_t \,\ket{\Phi_b},
\qquad
f_{A^{\iu}_{ab}}(t)=\frac{f_{P_a}(t)+f_{P_b}(t)}{2}-\Imag\bra{\Phi_a}\, J_t\, \ket{\Phi_b}.
\]
Moreover, the initial signal $f_A(0)=\abs{\tfrac{1}{d}\tr A}^2$ is known for every witness:
\[
f_I(0)=1,\qquad
f_{A_{a0}}(0)=f_{A_{0a}}(0)=f_{A^{\iu}_{a0}}(0)=f_{A^{\iu}_{0a}}(0)=\frac{1}{2}
\quad(a\neq0),
\]
and $f_A(0)=0$ for every other witness $A\in\cW$.

\end{lemma}
\begin{proof}
Because the map $A\mapsto(I\otimes A)\ket{\Phi_0}=\ket{\Phi_A}$ is linear, the Bell vectors associated with the two pair witnesses are
\[
\ket{\Phi_{A_{ab}}} = \frac{\ket{\Phi_a} + \ket{\Phi_b}}{\sqrt{2}},
\qquad
\ket{\Phi_{A^{\iu}_{ab}}} = \frac{\ket{\Phi_a} + \iu\ket{\Phi_b}}{\sqrt{2}}.
\]
For convenience, write $z\coloneq\bra{\Phi_a}\,J_t\,\ket{\Phi_b}$. Since $J_t$ is Hermitian, $\bra{\Phi_b}\,J_t\,\ket{\Phi_a}=\overline z$. Expanding the first pair witness by definition gives
\begin{align*}
f_{A_{ab}}(t) &= \frac{1}{2} \ab\big(\bra{\Phi_a} + \bra{\Phi_b})\,J_t\,\ab\big(\ket{\Phi_a} + \ket{\Phi_b})\\
&= \frac{1}{2} \ab\Big(\bra{\Phi_a}\, J_t\, \ket{\Phi_a} + \bra{\Phi_b}\, J_t\, \ket{\Phi_b} + z + \overline z)\\
&=\frac{f_{P_a}(t)+f_{P_b}(t)}{2}+\Real z,
\end{align*}
using $z+\overline z=2\Real z$. Similarly, for the second pair witness, we have
\begin{align*}
f_{A^{\iu}_{ab}}(t) &= \frac{1}{2} \ab\big(\bra{\Phi_a}-\iu\bra{\Phi_b})\,J_t\,\ab\big(\ket{\Phi_a}+\iu\ket{\Phi_b})\\
&= \frac{1}{2} \ab\Big(\bra{\Phi_a}\, J_t\, \ket{\Phi_a} + \iu z -\iu\overline z - \iu^2 \bra{\Phi_b}\, J_t\, \ket{\Phi_b})\\
&= \frac{f_{P_a}(t) + f_{P_b}(t)}{2} - \Imag z,
\end{align*}
using $\iu z-\iu\overline z=\iu\ab\big(z-\overline z)=\iu\cdot2\iu\,\Imag z=-2\Imag z$.

For the initial signal at $t=0$, $J_0 = \ketbra{\Phi_0}{\Phi_0}$ gives
\[
f_A(0) = \braket{\Phi_A}{\Phi_0} \braket{\Phi_0}{\Phi_A} = \abs{\braket{\Phi_0}{\Phi_A}}^2,
\qquad
\braket{\Phi_0}{\Phi_A} = \bra{\Phi_0}(I\otimes A)\ket{\Phi_0} = \frac{1}{d} \tr A.
\]
Hence $f_A(0) = \abs{\tfrac{1}{d} \tr A}^2$. The initial signal for every witness follows directly from the fact that nonidentity Pauli operators are traceless and $\tr I = d$.
\end{proof}

Differentiating the witness signals at $t=0$ and substituting the Bell-basis representation of the generator yields the Lindbladian coefficients.

\begin{lemma}[Coefficient identities]\label{lem:witness_identities}
The first two identities below hold for every nonidentity Pauli label $a$.
The third holds for every pair of distinct nonidentity Pauli labels $a,b$.
\begin{align}
h_a &= f'_{A^{\iu}_{a0}}(0) - \frac{1}{2}\ab\big(f'_{P_a}(0) + f'_{I}(0)), \label{eq:witness_ham}\\
\gamma_{aa} &= f'_{P_a}(0),\label{eq:witness_diag}\\
\gamma_{ab} &= \sbra[\Big]{f'_{A_{ab}}(0) - \frac{1}{2}\ab\big(f'_{P_a}(0) + f'_{P_b}(0))}
+\iu\sbra[\Big]{\frac{1}{2}\ab\big(f'_{P_a}(0) + f'_{P_b}(0))-f'_{A^{\iu}_{ab}}(0)}.\label{eq:witness_offdiag}
\end{align}
\end{lemma}
\begin{proof}
Since $f_A(t) = \tr\ab\big(\Pi_A J_t)$ is a fixed linear functional of $J_t$, its derivative at $t=0$ is the same functional evaluated on $J'_0$:
\begin{equation}\label{eq:deriv_functional}
f'_A(0) = \tr\ab\big(\Pi_A J'_0) = \bra{\Phi_A} \,J'_0\, \ket{\Phi_A}.
\end{equation}

\emph{Diagonal identity.} For $A=P_a$ with $a\neq0$, combining \cref{eq:deriv_functional} with \cref{lem:bell_chi} and \cref{lem:chi_to_coeff} gives
\[
f'_{P_a}(0) = \bra{\Phi_a} \, J'_0 \, \ket{\Phi_a} = \chi_{aa} = \gamma_{aa}.
\]

\emph{Off-diagonal identity.} Fix distinct labels $a,b$, allowing $b=0$. The polarization identities of \cref{lem:witness_signals} hold for every $t\geq0$, and every term in them is differentiable. Since taking real and imaginary parts commutes with differentiation, differentiating the two identities at $t=0$ and rearranging yields
\begin{align}
\Real\bra{\Phi_a} \, J'_0 \,\ket{\Phi_b} &= f'_{A_{ab}}(0) - \frac{1}{2} \ab\big(f'_{P_a}(0) + f'_{P_b}(0))\label{eq:re_extract},\\
\Imag\bra{\Phi_a}\, J'_0 \, \ket{\Phi_b} &= \frac{1}{2} \ab\big(f'_{P_a}(0) + f'_{P_b}(0)) - f'_{A^{\iu}_{ab}}(0)\label{eq:im_extract}.
\end{align}
For $a,b \neq 0$, \cref{lem:bell_chi} and \cref{lem:chi_to_coeff} give $\bra{\Phi_a} \, J'_0 \, \ket{\Phi_b} = \chi_{ab} = \gamma_{ab}$. Substituting into \cref{eq:re_extract,eq:im_extract} and recombining the two parts,
\[
\gamma_{ab} = \Real\gamma_{ab} + \iu \Imag\gamma_{ab} =
\sbra[\Big]{f'_{A_{ab}}(0) - \frac{1}{2}\ab\big(f'_{P_a}(0) + f'_{P_b}(0))} + \iu\sbra[\Big]{\frac{1}{2}\ab\big(f'_{P_a}(0) + f'_{P_b}(0)) - f'_{A^{\iu}_{ab}}(0)}.
\]

\emph{Hamiltonian identity.} Apply \cref{eq:im_extract} to the boundary pair $(a,0)$, where $a\neq 0$, $A^{\iu}_{a0} = (P_a + \iu I)/\sqrt{2}$, and $f_{P_0} = f_I$. We obtain
\[
\Imag\bra{\Phi_a} \, J'_0 \, \ket{\Phi_0} = \frac{1}{2} \ab\big(f'_{P_a}(0) + f'_{I}(0)) - f'_{A^{\iu}_{a0}}(0).
\]
By \cref{lem:bell_chi}, $\bra{\Phi_a}\,J'_0\,\ket{\Phi_0}=\chi_{a0}$. Hermiticity of $J'_0$ gives $\chi_{a0}=\overline{\chi_{0a}}$, while the same lemma gives $\Imag\chi_{0a}=h_a$. Hence $\Imag\bra{\Phi_a}\,J'_0\,\ket{\Phi_0}=-h_a$. Substituting and solving for $h_a$,
\[
h_a = f'_{A^{\iu}_{a0}}(0) - \frac{1}{2}\ab\big(f'_{P_a}(0) + f'_{I}(0)).\qedhere
\]
\end{proof}

As in the support-learning procedure of \cref{sec:support_learning}, the Hamiltonian contribution appears only in the imaginary part of a boundary element. Support learning uses its magnitude, whereas the witnesses read it off directly. To estimate the derivatives $f'_A(0)$, we again use Chebyshev interpolation as described in \cref{sec:chebyshev_interpolation}. This requires a uniform bound on the time derivatives of the signals. Since every witness $A \in \cW$ satisfies $\norm{A}_\infty\leq\sqrt{2}$, \cref{lem:choi_derivative_bound} gives
\begin{equation}\label{eq:witness_derivative}
\abs[\big]{f_A^{(r)}(t)} \leq \norm{A}_\infty^2\,\Lambda^r\leq 2\Lambda^r
\qquad\text{for every } r \in \N \text{ and } t \geq 0.
\end{equation}

\subsection{Ancilla-free process shadows}\label{sec:process_shadow}

We use ancilla-free two-sided process shadows to estimate the Bell-projector signals $f_A(t)$ simultaneously~\cite{HKP20,LLC24}.
Write $\Cl_n$ for the $n$-qubit Clifford group, whose elements map Pauli operators to Pauli operators up to phase under conjugation, and $\Stab_n$ for the set of $n$-qubit pure stabilizer-state projectors $C\ketbra{0^n}{0^n}C^\dagger$ with $C\in\Cl_n$.
We use the exact low-order moments of the uniform stabilizer ensemble.

\begin{fact}[Stabilizer moments~\cite{KG15,Web16,Zhu17}]\label{fact:stab_design}
Let $\varphi$ be a uniformly random element of $\Stab_n$ and $d=2^n$. Then
\[
\E[\varphi]=\frac{I}{d},
\qquad
\E[\varphi\otimes\varphi]=\frac{I\otimes I+\SWAP}{d(d+1)},
\qquad
\E[\varphi\otimes\varphi\otimes\varphi] = \frac{\sum_{\pi\in S_3}W_\pi}{d(d+1)(d+2)},
\]
where $\SWAP$ exchanges the two factors and $W_\pi$ permutes the three factors according to $\pi$.
\end{fact}

That is, the stabilizer states form an exact complex projective $3$-design.
Fix a time $t\geq0$ and consider the following experiment with random Clifford probing.
\begin{enumerate}
\item Sample a uniformly random stabilizer state $\psi\in\Stab_n$;
\item Evolve under $e^{t\cL}$;
\item Apply an independent uniformly random Clifford $C\in\Cl_n$, measure in the computational basis to obtain an outcome $b\in\{0,1\}^n$, and set $\phi\coloneq C^\dagger\ketbra{b}{b}C$.
\end{enumerate}
Record the classical pair $(\psi,\phi)$ and define the \emph{snapshot}
\begin{equation}\label{eq:snapshot}
\hat J_t\coloneq\ab\big((d+1)\psi^{\top}-I)\otimes\ab\big((d+1)\phi-I).
\end{equation}
Here $\psi^\top=\overline\psi$ is again a stabilizer state, and $\hat J_t$ is a classical object that is never physically prepared and is used only in post-processing.

\begin{lemma}[Two-sided process shadow]\label{lem:shadow_unbiased}
$\E\ab[\hat J_t]=J_t$.
\end{lemma}
\begin{proof}
Condition on $\psi$ and let $\sigma\coloneq e^{t\cL}(\psi)$. For every fixed outcome $b$, the projector $C^\dagger\ketbra{b}{b}C$ is uniform over $\Stab_n$ if $C$ is uniform over $\Cl_n$. Thus for every function $g$,
\begin{equation}\label{eq:born_weighted}
\E\ab[g(\phi) \mid \psi]
=\E_{C} \sum_{b} \bra{b}C\sigma C^\dagger\ket{b} \cdot g\ab\big(C^\dagger\ketbra{b}{b}C)
= d \E_{\varphi}\ab[ \tr(\sigma\varphi) g(\varphi)],
\end{equation}
with $\varphi$ uniform over $\Stab_n$. Taking $g(\varphi)=\varphi$ in \cref{eq:born_weighted} and using
\[
\tr(\sigma\varphi)\,\varphi=\tr_1\ab((\sigma\otimes I)(\varphi\otimes\varphi)),
\]
where $\tr_1$ denotes the partial trace over the first subsystem, gives
\begin{align*}
    \E[\phi \mid \psi] = d \E_{\varphi} \ab[\tr(\sigma\varphi) \,\varphi] &= d \E_{\varphi}\ab[ \tr_1\ab((\sigma \otimes I)(\varphi \otimes \varphi))]\\
    &= d \tr_1\ab\Big((\sigma\otimes I)\, \E_\varphi[\varphi \otimes \varphi])\\
    &=  d \tr_1\ab\Big((\sigma\otimes I)\, \frac{I\otimes I + \SWAP}{d(d+1)}) \qquad \text{(by the second moment in \cref{fact:stab_design})}\\
    &=\frac{\tr(\sigma) \,I + \sigma}{d + 1} = \frac{I + \sigma}{d + 1}.
\end{align*}
By linearity of expectation, we have
\[
\E[(d+1) \phi - I \mid \psi] = \sigma = e^{t\cL}(\psi),
\]
which is the classical-shadow identity for random Clifford measurements~\cite{HKP20}. Hence
\begin{align}
\E\ab[\hat J_t] &= \E_\psi \ab[\ab\big((d+1) \psi^\top - I) \otimes \E \ab[(d+1)\phi-I\mid\psi]]\notag\\
&= \E_\psi \ab\big[\ab\big((d+1)\psi^\top - I)\otimes e^{t\cL}(\psi)]\notag\\
&= (\mathrm{id} \otimes e^{t\cL}) \ab\Big(\E_\psi \ab[\ab\big((d+1)\psi^\top-I) \otimes \psi])\label{eq:hat_Jt}.
\end{align}
For an operator $X$, let $X^{\top_1}$ denote its partial transpose on the first subsystem. Then
\[
\SWAP^{\top_1} = \sum_{j,l=0}^{d-1} \ab(\ketbra{j}{l})^{\top} \otimes \ketbra{l}{j} = \sum_{j,l=0}^{d-1}\ketbra{l}{j}\otimes\ketbra{l}{j} = d \ketbra{\Phi_0}{\Phi_0}.
\]
Combining this identity with the second moment in \cref{fact:stab_design} gives
\[
\E_\psi \ab[ \psi^\top \otimes \psi]
=\ab\Big(\E_\psi[\psi \otimes \psi])^{\top_1}
=\frac{(I\otimes I)^{\top_1} + \SWAP^{\top_1}}{d(d+1)}
=\frac{I \otimes I + d \ketbra{\Phi_0}{\Phi_0}}{d(d+1)}.
\]
The first moment in \cref{fact:stab_design} then gives
\begin{align*}
\E_\psi \ab\big[\ab\big((d+1) \psi^\top-I) \otimes \psi]
&= (d+1)\E_\psi \ab[\psi^\top \otimes \psi] - I \otimes \E_\psi[\psi]\\
&= \frac{I \otimes I + d\ketbra{\Phi_0}{\Phi_0}}{d} - \frac{I \otimes I}{d}\\
&= \ketbra{\Phi_0}{\Phi_0}.
\end{align*}
Using this identity in \cref{eq:hat_Jt} gives
\[
\E\ab[\hat J_t] = (\mathrm{id} \otimes e^{t\cL}) \ab\Big(\E_\psi \ab[\ab\big((d+1) \psi^\top - I) \otimes \psi]) = (\mathrm{id} \otimes e^{t\cL}) (\ketbra{\Phi_0}{\Phi_0}) = J_t.\qedhere
\]
\end{proof}

The Born probability contributes one stabilizer factor to the second-moment calculation, so the conditional expectation is determined by the third moment of the uniform stabilizer ensemble.
This gives the following dimension-independent bound.

\begin{lemma}[Unbiased shadow estimator with constant second moment]\label{lem:shadow_variance}
Let $A$ be a normalized operator with $\norm{A}_\infty \leq \sqrt{2}$, and let $X_A \coloneq \tr\ab\big(\Pi_A \hat J_t)$ be the single-snapshot estimator. Then
\[
\E[X_A] = f_A(t), \qquad \E\ab[X_A^2] \leq 160.
\]
\end{lemma}
\begin{proof}
Unbiasedness follows directly from \cref{lem:shadow_unbiased} and linearity of expectation:
\[
\E[X_A] = \E[\tr\ab\big(\Pi_A \hat J_t)] = \tr\ab\big(\Pi_A J_t) = f_A(t).
\]

Set $S(\psi) \coloneq(d+1) \psi - I$, so that $\hat J_t=S(\psi)^\top\otimes S(\phi)$, and define
$C \coloneq A\psi A^\dagger$,
$\alpha\coloneq\tr\ab\big(\psi A^\dagger A)$, and
$\beta\coloneq\tr\ab\big(AA^\dagger\phi)$. Expanding $X_A$ gives
\begin{align}
X_A &= \bra{\Phi_0}(I \otimes A^\dagger) \ab\big(S(\psi)^\top \otimes S(\phi))(I \otimes A) \ket{\Phi_0} = \frac{1}{d}\tr\ab\big(S(\psi) A^\dagger S(\phi)A) \notag \\
&= \frac{1}{d}\ab((d+1)^2 \tr\ab\big(\psi A^\dagger \phi A)
- (d+1) \tr\ab\big(\psi A^\dagger A)
- (d+1) \tr\ab\big(A^\dagger \phi A) + \tr\ab\big(A^\dagger A))\notag \\
&= \frac{1}{d}\ab((d+1)^2 \tr(C \phi) - (d+1) \alpha -(d+1)\beta + d) \notag\\
&= \frac{(d+1)^2}{d} \tr(C\phi) - \frac{d+1}{d} \ab\big(\alpha + \beta) + 1.\label{eq:X_A}
\end{align}
Since $\psi$ is pure, $C$ is positive semidefinite of rank at most one with $\tr C=\alpha$, $\norm{C}_\infty=\alpha$, and $C^2=\alpha C$. Moreover, $0\leq\alpha,\beta\leq\norm{A}_\infty^2\leq2$.

Conditioned on $\psi$, the outcome $\phi$ follows the Born-weighted stabilizer ensemble in \cref{eq:born_weighted} with $\sigma = e^{t\cL}(\psi)$. Taking $g(\varphi)=\tr(C\varphi)^2$,
\begin{align*}
\E\ab[\tr(C\phi)^2 \mid \psi] &= d \E_\varphi \ab[\tr(\sigma \varphi) \tr(C \varphi)^2]\\
&= d \tr\sbra[\Big]{(\sigma\otimes C \otimes C)\,\E_\varphi\ab[\varphi\otimes\varphi\otimes\varphi]}\\
&= \frac{\sum_{\pi\in S_3}\tr\ab[W_\pi\ab\big(\sigma\otimes C\otimes C)]}{(d+1)(d+2)} \qquad \text{(by the third moment of \cref{fact:stab_design}).}
\end{align*}
The trace factors over the cycles of $\pi$.
Since $\tr\sigma=1$, $\tr C=\alpha$, and $C^2=\alpha C$, the identity permutation and $(2\,3)$ each contribute $\alpha^2$.
The transpositions $(1\,2)$ and $(1\,3)$ each contribute $\tr C\,\tr(\sigma C)=\alpha\tr(\sigma C)$.
The two $3$-cycles each contribute $\tr(\sigma C^2)=\alpha\tr(\sigma C)$.
Summing the six traces gives
\[
\sum_{\pi\in S_3} \tr\ab[W_\pi\ab\big(\sigma\otimes C\otimes C)]
= 2 \alpha^2 + 4\alpha \tr(\sigma C).
\]
H\"older's inequality gives $\tr(\sigma C) \leq \norm{\sigma}_1 \norm{C}_\infty = \alpha$. Therefore,
\[
\E\ab[\tr(C\phi)^2 \mid \psi] = \frac{2\alpha^2 + 4\alpha \tr(\sigma C)}{(d+1)(d+2)} \leq \frac{6\alpha^2}{(d+1)(d+2)}.
\]
Multiplying by the squared prefactor in \cref{eq:X_A} and using $d \geq 2$ gives
\[
\ab\Big(\frac{(d+1)^2}{d})^{2} \E\ab[\tr(C\phi)^2\mid\psi]
\leq\frac{6\alpha^2 (d+1)^3}{d^2(d+2)}
\leq 6\alpha^2 \cdot \frac{27}{16} \leq 11\alpha^2.
\]
The remaining term satisfies the pointwise bound
\[
\abs[\Big]{1-\frac{d+1}{d}(\alpha + \beta)} \leq 6,
\]
since $0\leq \alpha, \beta\leq 2$ and $0\leq\frac{d+1}{d}(\alpha+\beta)\leq 6$ for $d\geq 2$. Combining the two parts with $(x+y)^2\leq2x^2+2y^2$,
\begin{align*}
\E\ab[X_A^2 \mid \psi ] &\leq 2 \ab\Big(\frac{(d+1)^2}{d})^{2} \E\ab[\tr(C\phi)^2\mid\psi]
+ 2 \E\ab[\abs[\Big]{1-\frac{d+1}{d}(\alpha + \beta)}^2 \mid \psi]\\
&\leq 2\cdot11\alpha^2+2\cdot6^2\leq160,
\end{align*}
which completes the proof by averaging over $\psi$.
\end{proof}

The estimator $X_A$ is real but not uniformly bounded, since individual snapshots can be as large as $O(d)$.
We therefore use median-of-means estimation as in~\cref{fact:mom}.
Since a witness enters only through the classical evaluation of $X_A$, the experiment itself does not depend on $A$.
A single collection of snapshots $(\psi_m,\phi_m)$ therefore serves all witnesses in the list simultaneously.

\subsection{The algorithm and guarantee}\label{sec:shadow_alg}

The full algorithm collects shadow snapshots at the Chebyshev--Lobatto nodes, estimates the witness signals using median of means, computes their endpoint derivatives, and assembles the coefficients using \cref{lem:witness_identities}.

\begin{algorithm}[htbp!]
\caption{Learning Lindbladian coefficients by Clifford probing}\label{alg:shadow}
\small
\begin{algorithmic}[1]
\Require accuracy $\eps$; failure probability $\delta$; dynamical strength bound $\Lambda$; candidate supports $\hat\cS_H$ and $\hat\cS_D$.
\Ensure estimates $\{\hat h_a\}_{a\in\hat\cS_H}$ and $\{\hat\gamma_{ab}\}_{(a,b)\in\hat\cS_D}$.
\State Collect the participating labels $\hat\cS\gets\hat\cS_H\cup\cbra{a:\exists b,\ (a,b)\in\hat\cS_D}\cup\cbra{b:\exists a,\ (a,b)\in\hat\cS_D}$ and form the witness list
\NoNumState $\cW \gets\{I\}\cup\cbra{P_a: a\in\hat\cS}\cup\cbra{A^{\iu}_{a0}: a\in\hat\cS_H}\cup\cbra{A_{ab},\,A^{\iu}_{ab}: (a,b)\in\hat\cS_D,\ a\neq b}$.
\State Set $T\gets1/\Lambda$ and $q\gets\min\cbra[\big]{q'\in\N_{>0}:(q'+1)!\geq12\Lambda/\eps}$.
\State Compute the nodes $t_j$ and the endpoint weights $\ell'_j(0)$ for $j=0,1,\ldots,q$, and set $W\gets\sum_{j=1}^q\abs{\ell'_j(0)}$.
\State Set the node accuracy $\eta_0\gets\eps/(6W)$, the batch count $K\gets\ceil{8\log(2q\abs{\cW}/\delta)}$, and the shot count $N_{\mathrm{node}}\gets K\cdot\lceil640/\eta_0^2\rceil$.
\LineComment{Data collection:}
\For{$j=1,\ldots,q$}
  \For{$m=1,\ldots,N_{\mathrm{node}}$}
    \State Prepare a uniformly random stabilizer state $\psi_{j,m} \in \Stab_n$;
    \State Evolve under $e^{t_j\cL}$;
    \State Apply a uniformly random Clifford circuit $C\in\Cl_n$;
    \State Measure in the computational basis to obtain an outcome $b \in\{0,1\}^n$, and set $\phi_{j,m} \gets C^\dagger\ketbra{b}{b}C$.
  \EndFor
\EndFor
\LineComment{Classical post-processing:}
\For{each witness $A \in \cW$}
  \For{$j=1,\ldots,q$}
    \For{$m=1,\ldots,N_{\mathrm{node}}$}
      \State $\displaystyle X_A^{(j,m)} \gets \frac{1}{d} \tr\ab\Big(\ab\big((d+1)\psi_{j,m}-I) A^\dagger \ab\big((d+1)\phi_{j,m}-I) A)$.
    \EndFor
    \State Partition the $N_{\mathrm{node}}$ snapshot values into $K$ batches of equal size.
    \State Let $\hat f_A(t_j)$ be the median of the $K$ batch means.
  \EndFor
\EndFor
\For{each witness $A \in \cW$}
  \State $\displaystyle \hat f'_A(0)\gets\ell'_0(0) \abs[\big]{\frac{1}{d}\tr A}^2+\sum_{j=1}^q\ell'_j(0)\,\hat f_A(t_j)$.
\EndFor

\State Reconstruct the candidate coefficients:
\begin{align*}
\hat h_a &\gets \hat f'_{A^{\iu}_{a0}}(0)- \frac{1}{2}\ab\big(\hat f'_{P_a}(0)+\hat f'_{I}(0)), && a\in\hat\cS_H,\\
\hat\gamma_{aa} &\gets \hat f'_{P_a}(0), && (a,a)\in\hat\cS_D,\\
\hat\gamma_{ab} & \gets\ab\Big(\hat f'_{A_{ab}}(0) - \frac{1}{2}\ab\big(\hat f'_{P_a}(0)+\hat f'_{P_b}(0))) + \iu\ab\Big(\frac{1}{2}\ab\big(\hat f'_{P_a}(0)+\hat f'_{P_b}(0))-\hat f'_{A^{\iu}_{ab}}(0)), && (a,b)\in\hat\cS_D,\ a\neq b.
\end{align*}
\State \Return the estimates $\{\hat h_a\}$ and $\{\hat\gamma_{ab}\}$.
\end{algorithmic}
\end{algorithm}

\begin{theorem}[Coefficient learning by Clifford probing]\label{thm:shadow_full}
Let $\cL$ be an arbitrary $n$-qubit Lindbladian with $\norm{\cL^\dag}_{\infty\to\infty}\leq\Lambda$. Suppose candidate supports $\hat\cS_H$ and $\hat\cS_D$, of total size $M\coloneq\abs{\hat\cS_H}+\abs{\hat\cS_D}$, are given. For any $\eps,\delta\in(0,1)$, \cref{alg:shadow} outputs, with probability at least $1-\delta$, estimates satisfying
\[
\max_{a\in\hat\cS_H}\abs[\big]{\hat h_a-h_a} \leq \eps
\qquad \text{and} \qquad \max_{(a,b)\in\hat\cS_D}\abs[\big]{\hat\gamma_{ab}-\gamma_{ab}} \leq \eps.
\]
The total number of experiments is
\[
N=\widetilde O\ab\Big(\frac{\Lambda^2}{\eps^2}\log\frac{M}{\delta}).
\]
The total evolution time is $\Ttot=\widetilde O\ab\big(\Lambda\eps^{-2}\log(M/\delta))$, and the classical running time is $O\ab\big(NMn^3)$.
\end{theorem}
\begin{proof}
The algorithm evaluates only the witnesses in the list $\cW$, which contains the identity, one diagonal witness per participating label, one boundary-pair witness per Hamiltonian candidate, and at most two pair witnesses per dissipative candidate, so
\[
\abs{\cW}\leq 1+\ab\big(\abs{\hat\cS_H}+2\abs{\hat\cS_D})+\abs{\hat\cS_H}+2\abs{\hat\cS_D}\leq 4M+1.
\]

Fix $A \in \cW$ and decompose
\[
\hat f'_A(0)-f'_A(0)
= \underbrace{\ab\Big(\sum_{j = 0}^q\ell'_j(0)f_A(t_j)-f'_A(0))}_{\text{interpolation bias}}
+ \underbrace{\sum_{j=1}^q \ell'_j(0) \ab\big(\hat f_A(t_j)-f_A(t_j))}_{\text{statistical error}},
\]
which we bound separately.

For the interpolation bias, the rescaled signal $f_A/2$ satisfies the derivative bounds $\abs[\big]{(f_A/2)^{(r)}(t)}\leq\Lambda^r$ required by \cref{lem:interpolation_error}. The algorithm chooses $q$ so that $(q+1)!\geq12\Lambda/\eps$. Consequently, the interpolation bias of $f_A/2$ is at most $\eps/12$, and that of $f_A$ is at most $\eps/6$.

For the statistical error, \cref{lem:shadow_variance} gives $\E[X_A]=f_A(t_j)$ and $\E[(X_A)^2] \leq 160$. Applying \cref{fact:mom} with second-moment bound $160$ and accuracy $\eta_0$ shows that each median-of-means estimate satisfies
\[
\Pr\ab[\abs{\hat f_A(t_j)-f_A(t_j)}>\eta_0] \leq \frac{\delta}{q\abs{\cW}}.
\]
A union bound over the $q$ nodes and the $\abs{\cW}$ witnesses ensures that all $q\abs{\cW}$ estimates are simultaneously accurate with probability at least $1-\delta$. By the choice $\eta_0=\eps/(6W)$ of the algorithm,
\[
\abs[\Big]{\sum_{j=1}^q\ell'_j(0)\ab\big(\hat f_A(t_j)-f_A(t_j))}
\leq\eta_0\sum_{j=1}^q\abs{\ell'_j(0)}=\eta_0W\leq\frac{\eps}{6}.
\]
Combining the two parts, every witness derivative satisfies
\[
\abs[\big]{\hat f'_A(0)-f'_A(0)}\leq\frac{\eps}{6}+\frac{\eps}{6}=\frac{\eps}{3}.
\]

By \cref{lem:witness_identities}, each diagonal estimate $\hat\gamma_{aa}=\hat f'_{P_a}(0)$ has error at most $\eps/3$.
The real and imaginary parts of each off-diagonal estimate are combinations of three derivatives whose absolute weights sum to at most $2$, so each part is estimated with error at most $2\eps/3$, and
\[
\abs{\hat\gamma_{ab}-\gamma_{ab}}
\leq\sqrt{\ab\Big(\frac{2\eps}{3})^2+\ab\Big(\frac{2\eps}{3})^2}
=\frac{2\sqrt{2}}{3}\eps\leq\eps.
\]
Each Hamiltonian estimate is likewise a combination of three derivatives whose absolute weights sum to at most $2$, so its error is at most $2\eps/3\leq\eps$. The same bounds cover coefficients absent from $\cL$, because their true values are zero.

The number of experiments in \cref{alg:shadow} is
\[
N = q N_{\mathrm{node}} = O\ab\Big(q\cdot\frac{W^2}{\eps^2} \log\frac{q\abs{\cW}}{\delta}) = O\ab\Big(\frac{q^5\Lambda^2}{\eps^2}\log\frac{q\abs{\cW}}{\delta})
= \widetilde O\ab\Big(\frac{\Lambda^2}{\eps^2}\log \frac{M}{\delta}),
\]
using $W=O(q^2\Lambda)$ from \cref{lem:weight_bound}, $q=\widetilde\Theta(1)$, and $\abs{\cW}\leq4M+1$. Every experiment evolves for time at most $1/\Lambda$, giving $\Ttot\leq N/\Lambda$.

For the classical cost, although the snapshot $\hat J_t$ is written as a $4^n\times4^n$ matrix, the algorithm never computes it. Each snapshot is stored as the classical pair $(\psi_{j,m},\phi_{j,m})$, consisting of two stabilizer tableaux with $O(n^2)$ bits each~\cite{AG04}. Thus all snapshots occupy $O(Nn^2)$ bits. The post-processing needs only the scalars
\[
X_A^{(j,m)}=\frac{1}{d}\tr\ab\big(S(\psi_{j,m})\,A^\dagger S(\phi_{j,m})A),
\qquad
S(\xi)=(d+1)\xi-I,
\]
in the trace form established in the proof of \cref{lem:shadow_variance}. Since every witness $A$ is a combination of at most two Pauli operators, each such scalar expands into at most four terms of the form $\frac{1}{d}\tr\ab\big(S(\psi)P\,S(\phi)P')$ with Pauli operators $P,P'$. The required classical quantities are:
\begin{itemize}
    \item the overlap $\tr(\psi P\phi P')$, including its complex phase, computed from stabilizer inner products with Pauli insertions in $O(n^3)$ time~\cite{AG04}.
    \item the Pauli expectations $\tr(\psi\,PP')$ and $\tr(\phi\,P'P)$, where $PP'$ and $P'P$ are single Pauli operators up to the known phases described in \cref{fact:xi_ab}, each computable in $O(n^2)$ time.
    \item the known constant $\tr(PP')=d\,\mathbf{1}[P=P']$ by \cref{fact:pauli_inprod}.
\end{itemize}
Each of these quantities is zero or an integer power of $\sqrt{2}$ times an eighth root of unity, so the snapshot values are computed without expanding any matrices.
Evaluating one witness on one snapshot costs $O(n^3)$ time, and the total classical running time is $O(N\abs{\cW}\,n^3)=O(NMn^3)$.
\end{proof}

The candidate supports may contain every nonidentity label and every ordered pair of nonidentity labels, giving $M\leq4^n+16^n$ and $\log M=O(n)$.
Thus $\widetilde O(\Lambda^2\eps^{-2}(n+\log(1/\delta)))$ experiments suffice even for this full family, although the classical post-processing and output size are exponential.
The support-learning procedure of~\cref{sec:support_learning} supplies polynomial-size candidate supports.
Composing the two algorithms therefore gives the following recovery guarantee.

\begin{corollary}[Learning arbitrary Lindbladians in polynomial time]\label{cor:learning_arbitrary_lind}
Let $\cL$ be an arbitrary $n$-qubit Lindbladian with $\norm{\cL^\dag}_{\infty\to\infty}\leq\Lambda$. For any $\eps,\delta\in(0,1)$, there exists an algorithm that outputs coefficient estimates satisfying
\[
\max_{a\neq 0}{\abs[\big]{\hat h_a-h_a}} \leq \eps,
\qquad
\max_{\substack{a\neq0\\b\neq0}}{\abs[\big]{\hat\gamma_{ab}-\gamma_{ab}}} \leq \eps,
\]
with probability at least $1-\delta$.
On the same event, at most $9\Lambda^2/\eps^2$ estimates are nonzero, and every coefficient of magnitude at least $\eps$ has a nonzero estimate.
The total number of experiments is
\[
\widetilde O\ab\Big(\frac{\Lambda^2}{\eps^2}\log\frac{1}{\delta}),
\]
the total evolution time is
\[
\widetilde O\ab\Big(\frac{\Lambda}{\eps^2}\log\frac{1}{\delta}),
\]
and the classical running time is $\widetilde O(n^3\Lambda^{10}\eps^{-10}\log^5(1/\delta))$.
\end{corollary}
\begin{proof}
\emph{Step 1. Support learning.}
Run \cref{alg:displacement} with threshold $\eta\coloneq\eps$, failure probability $\delta/2$, and the dynamical strength bound $\Lambda$.
It returns the candidate supports $\hat\cS_H=\hat\cS$ and $\hat\cS_D=\hat\cS\times\hat\cS$.
By \cref{thm:support_learning}, with probability at least $1-\delta/2$, $\hat\cS_H$ contains every Hamiltonian label $a$ with $\abs{h_a}\geq\eps$ and $\hat\cS_D$ contains every dissipative pair $(a,b)$ with $\abs{\gamma_{ab}}\geq\eps$.

\emph{Step 2. Candidate size.}
The returned supports constitute a candidate family of total size
\[
M \coloneq \abs{\hat\cS_H} + \abs{\hat\cS_D} = \abs{\hat\cS}+\abs{\hat\cS}^2 = \widetilde O\ab\big(\Lambda^8 \eps^{-8} \log^4(1/\delta)).
\]

\emph{Step 3. Coefficient learning.}
Run \cref{alg:shadow} with fresh data, accuracy $\eps/3$, failure probability $\delta/2$, the dynamical strength bound $\Lambda$, and the candidate supports $(\hat\cS_H,\hat\cS_D)$.
Conditioned on the candidate supports, \cref{thm:shadow_full} ensures that, with probability at least $1-\delta/2$, the raw estimates satisfy
\begin{align}
\max_{a\in\hat\cS_H}\abs{\tilde h_a-h_a}\leq\frac{\eps}{3},
\qquad
\max_{(a,b)\in\hat\cS_D}\abs{\tilde\gamma_{ab}-\gamma_{ab}}\leq\frac{\eps}{3}.
\end{align}

\emph{Step 4. Output pruning.}
For each candidate coordinate, output its raw estimate if its magnitude is at least $2\eps/3$.
Assign zero to all other coordinates, and denote the resulting estimates by $\hat h_a$ and $\hat\gamma_{ab}$.
The output is the list of nonzero estimates together with their labels.

By a union bound, Steps 1 and 3 succeed simultaneously with probability at least $1-\delta$.
On this event, every coefficient of magnitude at least $\eps$ belongs to the candidate family and has a raw estimate of magnitude at least $\eps-\eps/3=2\eps/3$, so it appears in the output.
Every coefficient assigned zero therefore has magnitude less than $\eps$, while every nonzero estimate has error at most $\eps/3$.
This proves the entrywise error guarantee.
Moreover, every nonzero estimate corresponds to a true coefficient of magnitude at least $2\eps/3-\eps/3=\eps/3$.
If $K_H$ and $K_D$ count the nonzero Hamiltonian and dissipative estimates, respectively, \cref{lem:effective_sparsity} gives
\begin{align}
K_H+K_D\leq2K_H+K_D\leq\frac{9\Lambda^2}{\eps^2}.
\end{align}

For the complexity of the algorithm, Step 1 uses
\[
N_1 = 2N_{\mathrm{supp}}=\widetilde O\ab\Big(\frac{\Lambda^2}{\eps^2}\log\frac{1}{\delta})
\]
experiments and $\widetilde O\ab\big(n\Lambda^8\eps^{-8}\log^4(1/\delta))$ classical running time by \cref{thm:support_learning}. For Step 3, the candidate size gives $\log(M/\delta)=O\ab\big(\log(\Lambda/\eps)+\log(1/\delta))$, so by \cref{thm:shadow_full} the number of experiments is
\[
N_3=\widetilde O\ab\Big(\frac{\Lambda^2}{\eps^2}\log\frac{M}{\delta})=\widetilde O\ab\Big(\frac{\Lambda^2}{\eps^2}\log\frac{1}{\delta})
\]
and the classical running time is $O(N_3M\,n^3)=\widetilde O\ab\big(n^3\Lambda^{10}\eps^{-10}\log^5(1/\delta))$. Therefore, the total number of experiments is
\[
N = N_1 + N_3 = \widetilde O\ab\Big(\frac{\Lambda^2}{\eps^2}\log\frac{1}{\delta}).
\]
Every experiment of both stages evolves for time at most $1/\Lambda$, so
\[
\Ttot \leq \frac{N}{\Lambda} = \widetilde O\ab\Big(\frac{\Lambda}{\eps^2}\log\frac{1}{\delta}).
\]
Step 4 adds $O(M)$ scalar comparisons and at most $O(nM)$ time to copy the output labels.
The total classical running time is therefore $\widetilde O(n^3\Lambda^{10}\eps^{-10}\log^5(1/\delta))$, dominated by processing all $M$ candidate coordinates in Step 3.
The quadratic bound on the final output size does not reduce this intermediate computation.
\end{proof}

For a $k$-local Lindbladian with fixed $k$, the family of all $k$-local coordinates gives a polynomial-size candidate support that can be constructed without learning the nonzero support. Applying \cref{thm:shadow_full} directly to this family gives the following guarantee.

\begin{corollary}[Generic $k$-local Lindbladian learning]\label{cor:learning_klocal_lind}
Fix an integer $k\geq1$ with $k=O(1)$, and let $\cL$ be a $k$-local $n$-qubit Lindbladian with local dynamical strength at most $\Lloc$. Given $k$, $\Lloc$, and access to $\{e^{t\cL}\}_{t\geq0}$, for any $\eps,\delta\in(0,1)$ there is an algorithm that outputs coefficient estimates satisfying
\[
\max_{a\neq0}\abs[\big]{\hat h_a-h_a}\leq\eps,
\qquad
\max_{\substack{a\neq0\\b\neq0}}\abs[\big]{\hat\gamma_{ab}-\gamma_{ab}}\leq\eps
\]
with probability at least $1-\delta$. The algorithm is nonadaptive, ancilla-free, and control-free, and does not require knowledge of the nonzero support. Its experiment count and total evolution time satisfy
\[
\Nexp=\widetilde O\ab\Big(\frac{n^2\Lloc^2}{\eps^2}\log\frac{n}{\delta}),
\qquad
\Ttot=\widetilde O\ab\Big(\frac{n\Lloc}{\eps^2}\log\frac{n}{\delta}),
\]
and its classical running time is $\widetilde O\ab\big(n^{k+5}\Lloc^2\eps^{-2}\log(n/\delta))$.
\end{corollary}
\begin{proof}
Apply \cref{thm:shadow_full} to the $M=O(n^k)$ candidate coordinates consisting of all $k$-local Hamiltonian and dissipative coordinates, with $\Lambda=n\Lloc$ by \cref{eq:lb_global_general}.
\end{proof}

\section{Lower bounds for local Lindbladian learning}\label{sec:lb_local}

We prove lower bounds of $\Omega(n^2\Lloc^2/\eps^2)$ on the worst-case experiment count and $\Omega(n\Lloc/\eps^2)$ on the worst-case expected total evolution time. Both follow from learning the sign of one Hamiltonian coefficient in a family of $3$-local Lindbladians with local strength at most $\Lloc$, and hence also apply to learning all Hamiltonian and dissipative coefficients. The construction is motivated by the hard instances of~\cite{LTW26}.

\subsection{A family of $3$-local Lindbladians}\label{sec:lb_family}

Let $m=n-1\geq2$, with central qubit $0$ and bath qubits $1,\ldots,m$. Define
\[
\cQ\coloneq\{Z_i:i\in[m]\} \cup\{Z_iZ_j:1\leq i<j\leq m\}.
\]
For $r>0$ and $h\in\bbR$, assign rates to the singleton and pair labels by
\begin{equation*}
\begin{aligned}
r_i&=r,& u_i&=\frac{2h}{m},\\
r_{ij}&=\frac{r}{m-1},& u_{ij}&=-\frac{2h}{m(m-1)}.
\end{aligned}
\end{equation*}
For $Q\in\cQ$, define
\begin{equation*}
J_Q(h)\coloneq\sqrt{r_Q}Q +\iu\frac{u_Q}{\sqrt{r_Q}}Z_0Q.
\end{equation*}
Writing $\cD[A](\rho)=A\rho A^\dag-\tfrac12\{A^\dag A,\rho\}$, define
\begin{equation}\label{eq:lb_family}
\cL_h(\rho)\coloneq-\iu h[Z_0,\rho] +\sum_{Q\in\cQ}\cD[J_Q(h)](\rho) +r\cD[Z_0](\rho).
\end{equation}
The rates satisfy
\begin{equation}\label{eq:lb_cancellation}
\sum_{Q\in\cQ}u_Q=h, \qquad u_i+\sum_{j\ne i}u_{ij}=0.
\end{equation}
We also define
\[
w_Q\coloneq\frac{u_Q^2}{r_Q}, \qquad w_{\mathrm{tot}}\coloneq\sum_Qw_Q=\frac{6h^2}{mr}.
\]

\begin{lemma}\label{lem:lb_strength}
Let $n\geq3$, $\Lloc>0$, $r=\Lloc/16$, and $|h|\leq r/64$. The map $\cL_h$ in~\cref{eq:lb_family} is a $3$-local Lindbladian whose canonical Hamiltonian is $hZ_0$. Moreover,
\begin{equation}\label{eq:lb_strength_bounds}
\begin{aligned}
\frac{\Lloc}{4}&\leq\max_i\sum_{S\ni i}\norm{(\cL_h)_S^\dag}_{\infty\to\infty}\leq\Lloc,\\
\frac{n\Lloc}{8}&\leq\norm{\cL_h^\dag}_{\infty\to\infty}\leq n\Lloc.
\end{aligned}
\end{equation}
\end{lemma}

\begin{proof}
Each jump acts on the center and at most two bath qubits, so \cref{eq:lb_family} is a $3$-local Lindbladian. Its jumps are traceless, so expanding them in the traceless Pauli basis leaves the canonical Hamiltonian $hZ_0$.

\emph{Local strength.} Since $Z_0$ and $Q$ commute,
\[
\cD[J_Q(h)](\rho) =r_Q\cD[Q](\rho)+w_Q\cD[Z_0Q](\rho) +\iu u_Q[Z_0,Q\rho Q].
\]
The first term has support $\supp(Q)$, and the other two have union support $\{0\}\cup\supp(Q)$. Their adjoint norms are bounded by $2r_Q$, $2w_Q$, and $2|u_Q|$, respectively. Also, $\norm{[Z_0,\cdot]}_{\infty\to\infty}=2$ and $\norm{\cD[Q]^\dag}_{\infty\to\infty} = 2$. At the center, using $\sum_Q|u_Q|=3|h|$ and $\sum_Qw_Q=6h^2/(mr)$ gives
\[
\sum_{S\ni0}\norm{\cL_S^\dag}_{\infty\to\infty} \leq2r+2|h|+2\sum_Q(w_Q+|u_Q|) =2r+8|h|+\frac{12h^2}{mr}.
\]
At a bath qubit $i$, the singleton and $m-1$ pair labels give
\begin{align*}
\sum_{Q:i\in\supp(Q)}r_Q &=r+(m-1)\frac{r}{m-1}=2r,\\
\sum_{Q:i\in\supp(Q)}|u_Q| &=\frac{2|h|}{m}+(m-1)\frac{2|h|}{m(m-1)} =\frac{4|h|}{m},\\
\sum_{Q:i\in\supp(Q)}w_Q &=\frac{4h^2}{m^2r}+(m-1)\frac{4h^2}{m^2(m-1)r} =\frac{8h^2}{m^2r}.
\end{align*}
Consequently,
\[
\sum_{S\ni i}\norm{\cL_S^\dag}_{\infty\to\infty} \leq2\sum_{Q:i\in\supp(Q)}(r_Q+|u_Q|+w_Q) =4r+\frac{8|h|}{m}+\frac{16h^2}{m^2r}.
\]
For $m\geq2$ and $|h|\leq r/64$, both bounds are less than $16r=\Lloc$. Conversely, the distinct bath-only terms $r\cD[Z_i]$ and $r\cD[Z_iZ_j]/(m-1)$ already contribute $2r+(m-1)2r/(m-1)=4r$ at qubit $i$. Thus the local dynamical strength lies between $\Lloc/4$ and $\Lloc$.

\emph{Global strength.} Define
\begin{equation}\label{eq:lb_total_rate}
R\coloneq r+\sum_Q(r_Q+w_Q)=r+\frac32mr+\frac{6h^2}{mr}.
\end{equation}
The jump operators satisfy
\[
(\sqrt rZ_0)^\dag(\sqrt rZ_0)+\sum_Q J_Q(h)^\dag J_Q(h)=RI.
\]
Let $\ket{\psi}=\ket{+}^{\otimes n}$ and $A=I-2\ketbra{\psi}{\psi}$, so $\norm{A}_\infty=1$. Every jump operator
$J\in\{\sqrt rZ_0\}\cup\{J_Q(h):Q\in\cQ\}$ maps $\ket{\psi}$ to a state orthogonal to it. Hence, $A\ket{\psi}=-\ket{\psi}$ and $AJ\ket{\psi}=J\ket{\psi}$, which give
\begin{align*}
\bra{\psi}\cD[J]^\dag(A)\ket{\psi} &=\bra{\psi}J^\dag AJ\ket{\psi} -\frac12\bra{\psi}\{J^\dag J,A\}\ket{\psi}\\
&=2\bra{\psi}J^\dag J\ket{\psi}.
\end{align*}
The Hamiltonian contribution has expectation $\bra{\psi}\iu[hZ_0,A]\ket{\psi}=0$. Summing the jump contributions yields
\begin{align*}
\norm{\cL_h^\dag}_{\infty\to\infty} &\geq\bra{\psi}\cL_h^\dag(A)\ket{\psi} =2r+2\sum_Q(r_Q+w_Q)=2R\\
&\geq(3n-1)r\geq\frac{n\Lloc}{8}.
\end{align*}
The upper bound follows from~\cref{eq:lb_global_general} and the local strength bound.
\end{proof}

\subsection{A common overlap bound}\label{sec:lb_overlap}

The following lemma provides the overlap of a single experiment for both lower bounds on experiment count and total evolution time.

\begin{lemma} \label{lem:lb_uniform_overlap}
Let $\cL_{\pm h}$ be the family in~\cref{eq:lb_family}, where $m\geq2$, $r>0$, and $0<h\leq r/64$. For every $t\geq0$, there exist Stinespring isometries $V_{+,t}$ and $V_{-,t}$ with environment $E$ such that
\[
e^{t\cL_{\pm h}}(\rho)=\tr_E\ab\big(V_{\pm,t}\rho V_{\pm,t}^\dag)
\]
and
\begin{equation}\label{eq:lb_uniform_overlap}
V_{+,t}^\dag V_{-,t}=s(t)I,\qquad 0\leq1-s(t)\leq\min\{\beta,\kappa t\},
\end{equation}
where
\[
\beta\coloneq\frac{512h^2}{m(m-1)r^2}\leq\frac1{16}, \qquad \kappa\coloneq2w_{\mathrm{tot}}=\frac{12h^2}{mr}.
\]
\end{lemma}

The channel decomposition below expresses the evolution as a mixture of central-qubit dephasing channels and bath Pauli errors.
The associated scalar overlap gives the bound linear in $t$.
To obtain a bound uniform in $t$, we use cancellation between the singleton and pair terms to control a weighted squared difference.

Throughout these auxiliary results, fix $m\geq2$, $r>0$, $0<h\leq r/64$, and $t>0$, with the rates in~\cref{eq:lb_family}. For $a\in\F_2^m$, let $Z_{a}$ be the bath Pauli string that acts as $Z$ on qubit $i$ when the $i$th bit of $a$ is $1$, and as identity otherwise. Write $\mathcal Z_{a}(X)=Z_{a}XZ_{a}$. We denote the central qubit dephasing channel
\[
\Phi_z:\quad
\begin{pmatrix}\rho_{00}&\rho_{01}\\\rho_{10}&\rho_{11}\end{pmatrix}
\longmapsto
\begin{pmatrix}\rho_{00}&z\rho_{01}\\\overline z\rho_{10}&\rho_{11}\end{pmatrix},
\qquad |z|\leq1.
\]
For $Q\in\cQ$ with binary bath label $b_Q\in\F_2^m$, define operators on functions $f:\F_2^m\to\mathbb C$ by
\[
(T_Qf)(a)=f(a\oplus b_Q),\qquad D_Q=T_Q-I,\qquad D_i=D_{Z_i}.
\]
Thus $T_Q$ adds the bath error $Q$ to the current label. Let $\delta_0$ be the point mass at the zero bath label $(0,\ldots,0)\in\F_2^m$. Set
\[
r'_Q=r_Q-w_Q>0,\qquad C=\sum_Q\frac{u_Q}{h}D_Q.
\]
The positivity follows from $w_Q/r_Q=(u_Q/r_Q)^2$ and $|u_Q|/r_Q=2h/(mr)<1$.

\begin{lemma}[Channel decomposition]\label{lem:lb_channel_decomposition}
The channels $e^{t\cL_{\pm h}}$ admit the decompositions
\[
e^{t\cL_h}=\sum_{a}p_t(a)(\Phi_{z_{a}}\otimes\mathcal Z_{a}),
\qquad
e^{t\cL_{-h}}=\sum_{a}p_t(a)(\Phi_{\overline{z_{a}}}\otimes\mathcal Z_{a}),
\]
where $z_{a}=q_t(a)/p_t(a)$, with bath-label coefficients
\begin{align}
p_t&=\exp\ab\Big(t\sum_Q(r_Q+w_Q)D_Q)\delta_0, \label{eq:lb_p}\\
q_t&=e^{-2rt-2w_{\mathrm{tot}}t}
\exp\ab\Big(t\sum_Qr'_QD_Q)e^{2\iu htC}\delta_0. \label{eq:lb_q}
\end{align}
The coefficients $p_t(a)$ form a probability distribution, and $q_t$ satisfies
\begin{equation}\label{eq:lb_pq_properties}
\begin{gathered}
|q_t(a)|\leq e^{-2rt}p_t(a),\qquad \sum_{a}q_t(a)=e^{-2rt-2w_{\mathrm{tot}}t}.
\end{gathered}
\end{equation}
\end{lemma}

\begin{proof}
Write the input in the central $Z_0$ eigenbasis, with bath blocks $\rho_{cd}$ for $c,d\in\{0,1\}$. A jump $J_Q(h)$ acts on diagonal and upper off-diagonal blocks as
\begin{align*}
\cD[J_Q(h)](\rho)_{cc}&=(r_Q+w_Q)(Q\rho_{cc}Q-\rho_{cc}),\\
\cD[J_Q(h)](\rho)_{01}&=(r'_Q+2\iu u_Q)Q\rho_{01}Q-(r_Q+w_Q)\rho_{01}.
\end{align*}
Adding the Hamiltonian and central dephasing terms, and using $\sum_Q u_Q=h$, gives
\[
\cL_h(\rho)_{01}
=\sum_Q(r'_Q+2\iu u_Q)(Q\rho_{01}Q-\rho_{01})
-(2r+2w_{\mathrm{tot}})\rho_{01}.
\]
To exponentiate these block generators, fix a bath operator $X$ and consider $X_t=\sum_a f_t(a)Z_aXZ_a$, where $f_t:\F_2^m\to\mathbb C$ and $f_0=\delta_0$, so $X_0=X$. Since $QZ_a=Z_{a\oplus b_Q}$, reindexing gives
\begin{align*}
QX_tQ
&=\sum_a f_t(a)Z_{a\oplus b_Q}XZ_{a\oplus b_Q}\\
&=\sum_a f_t(a\oplus b_Q)Z_aXZ_a,\\
QX_tQ-X_t&=\sum_a(D_Qf_t)(a)Z_aXZ_a.
\end{align*}
Thus the diagonal and upper off-diagonal block equations are solved by coefficient functions satisfying, respectively,
\begin{align*}
\dot p_t&=\sum_Q(r_Q+w_Q)D_Qp_t,\\
\dot q_t&=\ab[\sum_Q(r'_Q+2\iu u_Q)D_Q-(2r+2w_{\mathrm{tot}})I]q_t,
\end{align*}
with $p_0=q_0=\delta_0$. The first equation gives~\cref{eq:lb_p}. For the second, the translations $T_Q$ commute, hence, so do the $D_Q$, and $\sum_Qu_QD_Q=hC$. Therefore,
\begin{align*}
q_t
&=e^{-2rt-2w_{\mathrm{tot}}t}
\exp\ab\Big(t\sum_Q(r'_Q+2\iu u_Q)D_Q)\delta_0\\
&=e^{-2rt-2w_{\mathrm{tot}}t}
\exp\ab\Big(t\sum_Qr'_QD_Q)
\exp\ab\Big(2\iu t\sum_Qu_QD_Q)\delta_0\\
&=e^{-2rt-2w_{\mathrm{tot}}t}
\exp\ab\Big(t\sum_Qr'_QD_Q)e^{2\iu htC}\delta_0,
\end{align*}
which is~\cref{eq:lb_q}. By the block evolution with specified initial value, we have
\begin{align*}
(e^{t\cL_h}(\rho))_{cc}&=\sum_a p_t(a)Z_a\rho_{cc}Z_a,\\
(e^{t\cL_h}(\rho))_{01}&=\sum_a q_t(a)Z_a\rho_{01}Z_a.
\end{align*}
Replacing $h$ by $-h$ changes the signs of the $u_Q$ and leaves $r_Q,w_Q,r'_Q$ unchanged, so it preserves $p_t$. Since the $D_Q$ and $\delta_0$ are real,
\[
\overline{q_t}
=e^{-2rt-2w_{\mathrm{tot}}t}
\exp\ab\Big(t\sum_Q(r'_Q-2\iu u_Q)D_Q)\delta_0
=q_t\big|_{h\mapsto-h}.
\]

These block formulas yield the claimed channel decompositions once~\cref{eq:lb_pq_properties} is established. Indeed, these properties make $p_t$ a probability distribution and ensure that $z_a=q_t(a)/p_t(a)$ is well defined with
\[
|z_a|=\frac{|q_t(a)|}{p_t(a)}\leq e^{-2rt}<1.
\]
Thus each $\Phi_{z_a}$ is a channel, and its weighted contribution is
\[
p_t(a)(\Phi_{z_a}\otimes\mathcal Z_a)(\rho)
=\begin{pmatrix}
p_t(a)Z_a\rho_{00}Z_a & q_t(a)Z_a\rho_{01}Z_a\\
\overline{q_t(a)}Z_a\rho_{10}Z_a & p_t(a)Z_a\rho_{11}Z_a
\end{pmatrix}.
\]
Summing over $a$ reproduces the block evolution above, with the lower off-diagonal block obtained by taking the adjoint of the upper one, and gives the decomposition of $e^{t\cL_h}$. The decomposition of $e^{t\cL_{-h}}$ follows from $q_t|_{h\mapsto-h}=\overline{q_t}$. These channel identities also remain valid after tensoring both sides with the identity map on an arbitrary reference system.

It remains to prove~\cref{eq:lb_pq_properties}. Set
\[
\lambda_Q\coloneq r_Q+w_Q>0,\qquad
\gamma_Q\coloneq r'_Q+2\iu u_Q,\qquad
R\coloneq\sum_Q\lambda_Q.
\]
Using $u_Q^2=r_Qw_Q$ and $\sum_Qu_Q=h$, we obtain
\begin{align*}
|\gamma_Q|^2
&=(r_Q-w_Q)^2+4u_Q^2=(r_Q+w_Q)^2=\lambda_Q^2,\\
\sum_Q\gamma_Q
&=\sum_Q(r_Q-w_Q)+2\iu h=R-2w_{\mathrm{tot}}+2\iu h.
\end{align*}
Since $D_Q=T_Q-I$ and the translations commute, the exponential series give
\begin{align*}
p_t
&=e^{-Rt}\prod_Q\ab(\sum_{k=0}^{\infty}\frac{(t\lambda_Q)^k}{k!}T_Q^k)\delta_0,\\
q_t
&=e^{-2rt-Rt-2\iu ht}\prod_Q\ab(\sum_{k=0}^{\infty}\frac{(t\gamma_Q)^k}{k!}T_Q^k)\delta_0.
\end{align*}
These series converge absolutely. Each $T_Q$ is a permutation matrix, so the triangle inequality and $|\gamma_Q|=\lambda_Q$ yield, for every bath label $a$,
\begin{align*}
|q_t(a)|
&\leq e^{-2rt-Rt}\ab[\prod_Q\ab(\sum_{k=0}^{\infty}\frac{(t|\gamma_Q|)^k}{k!}T_Q^k)\delta_0](a)\\
&=e^{-2rt}p_t(a).
\end{align*}
Every term in the expansion of $p_t$ is nonnegative. For $a=(a_1,\ldots,a_m)$, select $k=1$ in the factor for $Q=Z_i$ whenever $a_i=1$, and $k=0$ in every other factor. The resulting translations send $\delta_0$ to the point mass at $a$, so
\[
p_t(a)\geq e^{-Rt}\prod_{i:a_i=1}(t\lambda_{Z_i})>0,
\]
where we take the empty product to be $1$.

For any function $f:\F_2^m\to\mathbb C$, reindexing also gives
\[
\sum_a(D_Qf)(a)
=\sum_a f(a\oplus b_Q)-\sum_a f(a)=0.
\]
Summing the coefficient evolution equations yields
\[
\frac{\mathrm d}{\mathrm dt}\sum_a p_t(a)=0,\qquad
\frac{\mathrm d}{\mathrm dt}\sum_a q_t(a)
=-(2r+2w_{\mathrm{tot}})\sum_a q_t(a).
\]
Since both initial sums equal $\sum_a\delta_0(a)=1$, we have
\[
\sum_a p_t(a)=1,\qquad
\sum_a q_t(a)=e^{-2rt-2w_{\mathrm{tot}}t}.
\]
This proves~\cref{eq:lb_pq_properties} and completes the proof.
\end{proof}

\begin{lemma}[Overlap for dephasing channels]\label{lem:lb_scalar_overlap}
Let $(p_{a})_{a}$ be a positive probability distribution, and let $z_{a}=x_{a}+\iu y_{a}$ satisfy $|z_{a}|\leq\eta<1$ and $\sum_{a}p_{a}y_{a}=0$. The channels
\[
\mathcal E_+=\sum_{a}p_{a}(\Phi_{z_{a}}\otimes\mathcal Z_{a}),
\qquad
\mathcal E_-=\sum_{a}p_{a}(\Phi_{\overline{z_{a}}}\otimes\mathcal Z_{a})
\]
admit Stinespring isometries with a common environment such that
\begin{equation}\label{eq:lb_crossed_sum}
V_+^\dag V_-=(1-\Delta)I,\qquad
\Delta=\sum_{a}p_{a}\frac{y_{a}^2}{1+x_{a}}\geq0.
\end{equation}
For these same isometries,
\[
\Delta\leq\eta-\sum_{a}p_{a}x_{a},
\qquad
\Delta\leq8\sum_{a}p_{a}|z_{a}-\alpha_{a}|^2
\]
for every choice of nonnegative real numbers $\alpha_{a}$.
\end{lemma}

\begin{proof}
Fix a bath label $a$ and write $z=z_a=x+\iu y$. Choose
\[
B_0^\pm(z)=\sqrt{\frac{1+x}{2}}I
\ \pm\iu\frac{y}{\sqrt{2(1+x)}}Z_0,
\qquad
B_1^\pm(z)=\sqrt{\frac{1-|z|^2}{2(1+x)}}Z_0.
\]
Direct multiplication shows that these are Kraus representations of $\Phi_z$ and $\Phi_{\overline z}$, respectively, and that
\[
\sum_{\ell=0}^1(B_{\ell}^+(z))^\dag B_{\ell}^-(z)
=\ab\Big(1-\frac{y^2}{1+x})I-\iu yZ_0.
\]
For the full channels, use $K_{a,\ell}^\pm=\sqrt{p_{a}}B_{\ell}^\pm(z_{a})Z_{a}$ and attach the same orthonormal environment labels $\ket{a,\ell}$ to both Kraus families. Their crossed sum is
\[
V_+^\dag V_-=
\ab\Big(1-\sum_{a}p_{a}\frac{y_{a}^2}{1+x_{a}})I
-\iu\ab\Big(\sum_{a}p_{a}y_{a})Z_0.
\]
By $\sum_a p_a y_a=0$, the $Z_0$ term vanishes, proving~\cref{eq:lb_crossed_sum}.

Both bounds follow from pointwise inequalities for the summands of $\Delta$. First, $|z_{a}|\leq\eta$ implies
\[
\frac{y_{a}^2}{1+x_{a}}
\leq(\eta-x_{a})\frac{\eta+x_{a}}{1+x_{a}}
\leq\eta-x_{a}.
\]
Averaging with weights $p_{a}$ gives the first bound.

For the second bound, if $|z|\leq1$, $x>-1$, and $\alpha\geq0$, then
\begin{equation}\label{eq:lb_reflection_inequality}
\frac{y^2}{1+x}\leq8|z-\alpha|^2.
\end{equation}
Indeed, for $x\geq-1/2$, the left side is at most $2y^2\leq2|z-\alpha|^2$. Otherwise, it is at most $1-x\leq2$, while $|z-\alpha|^2\geq x^2>1/4$. Averaging with weights $p_{a}$ gives the second bound.
\end{proof}

\Cref{lem:lb_scalar_overlap} reduces the uniform bound to a weighted squared difference. For the next two lemmas, we use
\[
\norm{f}_\pi^2=\sum_{a}\frac{|f(a)|^2}{\pi(a)}
\]
for a probability distribution $\pi$, and write $\norm{A}_{\pi\to\pi}$ for the induced operator norm.

\begin{lemma}\label{lem:lb_product_estimate}
Let
\[
r'=r-\frac{4h^2}{m^2r},\qquad
b=\frac{1-e^{-2r't}}2,\qquad
\mu=\bigotimes_{i=1}^m(1-b,b),\qquad
v=\frac4{e^{4r't}-1}.
\]
Then the operator $C$ defined above satisfies
\begin{equation}\label{eq:lb_finite_exponential}
\norm{(e^{2\iu htC}-I)\mu}_\mu^2
\leq\frac{8h^2t^2v^2}{m(m-1)}e^{4ht(v+4)}.
\end{equation}
\end{lemma}

\begin{proof}
The identity $D_{Z_iZ_j}=D_iD_j+D_i+D_j$ rewrites $C$ as a sum of single-index and pair terms. By~\cref{eq:lb_cancellation}, the coefficient of each $D_i$ is $(u_i+\sum_{j\ne i}u_{ij})/h=0$. Thus
\begin{equation}\label{eq:lb_C}
C=-\frac{2}{m(m-1)}\sum_{i<j}D_iD_j.
\end{equation}
This cancellation allows us to use orthogonality between distinct pairs.

For the one-bit marginal $\mu_i=(1-b,b)$, we have $D_i\mu_i=(2b-1,1-2b)$ and
\[
\norm{D_i\mu_i}_{\mu_i}^2
=\frac{(1-2b)^2}{b(1-b)}=v.
\]
The vectors $\mu_i$ and $D_i\mu_i$ are orthogonal in the weighted inner product, since the entries of $D_i\mu_i$ sum to zero. Also, $D_i^2=-2D_i$, so in the orthonormal basis $\mu_i,D_i\mu_i/\sqrt v$,
\[
D_i=\begin{pmatrix}0&0\\\sqrt v&-2\end{pmatrix},
\qquad \norm{D_i}_{\mu_i\to\mu_i}^2=v+4.
\]
On the product distribution $\mu$, each $D_iD_j$ therefore has operator norm $v+4$. Applying the triangle inequality to~\cref{eq:lb_C} gives
\[
\norm{C}_{\mu\to\mu}\leq v+4.
\]
When $C$ acts on $\mu$ itself, we get a stronger estimate. For distinct pairs, $D_iD_j\mu$ are orthogonal. At a bit belonging to just one pair, the inner product contains the zero factor $\langle\mu_i,D_i\mu_i\rangle_{\mu_i}$. Each vector has squared norm $v^2$, and hence
\[
\norm{C\mu}_\mu^2
=\frac{4}{m^2(m-1)^2}\sum_{i<j}\norm{D_iD_j\mu}_\mu^2
=\frac{2v^2}{m(m-1)}.
\]
To retain this smaller factor while controlling the full exponential, use
\[
(e^{2\iu htC}-I)\mu
=2\iu ht\int_0^1e^{2\iu htsC}C\mu\,\mathrm ds.
\]
The operator-norm bound on $C$ controls the exponential, while the bound on $C\mu$ controls the vector inside the integral. Consequently,
\[
\begin{aligned}
\norm{(e^{2\iu htC}-I)\mu}_\mu^2
&\leq4h^2t^2e^{4ht\norm{C}_{\mu\to\mu}}\norm{C\mu}_\mu^2\\
&\leq\frac{8h^2t^2v^2}{m(m-1)}e^{4ht(v+4)},
\end{aligned}
\]
which proves~\cref{eq:lb_finite_exponential}.
\end{proof}

\begin{lemma}\label{lem:lb_coherence_estimate}
For $p_t,q_t$ from~\cref{lem:lb_channel_decomposition}, define the nonnegative reference
\[
p_{\mathrm b,t}=\exp\ab\Big(t\sum_Qr'_QD_Q)\delta_0,
\qquad q_{\mathrm{ref},t}=e^{-2rt-2w_{\mathrm{tot}}t}p_{\mathrm b,t}.
\]
Then
\begin{equation}\label{eq:lb_coherence_bound}
E_t\coloneq\sum_{a}\frac{|q_t(a)-q_{\mathrm{ref},t}(a)|^2}{p_t(a)}
\leq\frac{64h^2}{m(m-1)r^2}.
\end{equation}
\end{lemma}

\begin{proof}
The reference $q_{\mathrm{ref},t}$ is obtained from~\cref{eq:lb_q} by omitting $e^{2\iu htC}$. Thus
\[
q_t-q_{\mathrm{ref},t}
=e^{-2rt-2w_{\mathrm{tot}}t}(e^{2\iu htC}-I)p_{\mathrm b,t}.
\]
Since $r_Q+w_Q=r'_Q+2w_Q$, the distribution $p_t$ is obtained from $p_{\mathrm b,t}$ by adding independent errors of rate $2w_Q$. With probability $e^{-2w_{\mathrm{tot}}t}$, none of these extra errors occurs, giving
\begin{equation}\label{eq:lb_p_domination}
p_t(a)\geq e^{-2w_{\mathrm{tot}}t}p_{\mathrm b,t}(a).
\end{equation}
Substituting these facts into $E_t$ yields
\[
E_t\leq e^{-4rt-2w_{\mathrm{tot}}t}
\norm{(e^{2\iu htC}-I)p_{\mathrm b,t}}_{p_{\mathrm b,t}}^2.
\]

We now reduce this norm to the product distribution in~\cref{lem:lb_product_estimate}. For functions $f,g:\F_2^m\to\mathbb C$, define their convolution by
\[
(f*g)(a)\coloneq\sum_{\theta\in\F_2^m}g(\theta)f(a\oplus\theta).
\]
Let $\mu$ and $\nu$ be the bath-label distributions generated by the singleton and pair errors, respectively. These two error contributions are independent, and their labels combine by bitwise addition, so $p_{\mathrm b,t}=\mu*\nu$. The singleton rate is precisely $r'=r-4h^2/(m^2r)$, so $\mu$ is the distribution in that lemma. For any $f:\F_2^m\to\mathbb C$, Cauchy--Schwarz gives
\[
\frac{|\sum_{\theta}\nu(\theta)f(a\oplus \theta)|^2}
{\sum_{\theta}\nu(\theta)\mu(a\oplus \theta)}
\leq\sum_{\theta}\nu(\theta)\frac{|f(a\oplus \theta)|^2}{\mu(a\oplus \theta)}.
\]
Summing over $a$ and using $\sum_{\theta}\nu(\theta)=1$ proves the convolution contraction
\begin{equation}\label{eq:lb_convolution_contraction}
\norm{f*\nu}_{\mu*\nu}^2\leq\norm{f}_\mu^2.
\end{equation}
Translations commute with convolution, so $(e^{2\iu htC}-I)p_{\mathrm b,t}=((e^{2\iu htC}-I)\mu)*\nu$. Therefore,
\begin{equation}\label{eq:lb_product_reduction}
E_t\leq e^{-4rt-2w_{\mathrm{tot}}t}
\norm{(e^{2\iu htC}-I)\mu}_\mu^2.
\end{equation}
Applying~\cref{lem:lb_product_estimate} gives
\[
E_t\leq\frac{8h^2t^2v^2}{m(m-1)}
e^{4htv+(-4r+16h)t-2w_{\mathrm{tot}}t}.
\]
Finally, $tv\leq1/r'$ follows from $e^{4r't}-1\geq4r't$, and $-4r+16h\leq0$. Hence,
\[
E_t\leq\frac{8e^{4h/r'}h^2}{m(m-1)(r')^2}
\leq\frac{64h^2}{m(m-1)r^2},
\]
where the last inequality uses $r'\geq r/2$ and $e^{4h/r'}\leq e^{8h/r}<2$.
\end{proof}

\begin{proof}[Proof of \cref{lem:lb_uniform_overlap}]
Fix $t>0$ and take $p_t,q_t$ from~\cref{lem:lb_channel_decomposition}. Set $z_{a}=q_t(a)/p_t(a)$ and $\eta=e^{-2rt}$. By~\cref{eq:lb_pq_properties}, $|z_{a}|\leq\eta<1$ and $\sum_{a}p_t(a)\Imag z_{a}=0$. Thus~\cref{lem:lb_scalar_overlap} gives Stinespring isometries for $e^{t\cL_{\pm h}}$ with $V_{+,t}^\dag V_{-,t}=s(t)I$ and $\Delta(t)=1-s(t)\geq0$.

The first bound in~\cref{lem:lb_scalar_overlap}, together with the mass identity for $q_t$, yields
\begin{equation}\label{eq:lb_overlap_time_bound}
\Delta(t)\leq\eta-\Real\sum_{a}q_t(a)
=e^{-2rt}(1-e^{-2w_{\mathrm{tot}}t})
\leq2w_{\mathrm{tot}}t=\kappa t.
\end{equation}
For the second bound, choose $\alpha_{a}=q_{\mathrm{ref},t}(a)/p_t(a)\geq0$, using the reference in~\cref{lem:lb_coherence_estimate}. Applying the second bound in~\cref{lem:lb_scalar_overlap} and then~\cref{lem:lb_coherence_estimate} gives
\begin{equation}\label{eq:lb_overlap_reduction}
\Delta(t)\leq8\sum_{a}\frac{|q_t(a)-q_{\mathrm{ref},t}(a)|^2}{p_t(a)}
=8E_t\leq\frac{512h^2}{m(m-1)r^2}=\beta.
\end{equation}
Both bounds apply to the same isometries, proving~\cref{eq:lb_uniform_overlap}. Finally, $h\leq r/64$ and $m\geq2$ imply $\beta\leq1/16$.
\end{proof}

\subsection{Lower bound on experiment count}\label{sec:lb_count}

We use the bound $1-s(t)\leq\beta$ from \cref{lem:lb_uniform_overlap} to bound the worst-case experiment budget.

\begin{theorem}[Lower bound on experiment count]\label{thm:lb_count}
Let $k\geq3$, $\Lloc>0$, and $0<\eps\leq\Lloc/2048$. Suppose an algorithm estimates the Hamiltonian coefficients of every $k$-local $n$-qubit Lindbladian with local dynamical strength at most $\Lloc$ to entrywise error $\eps$, with success probability at least $2/3$. Then
\begin{equation}\label{eq:lb_count_theorem}
\Nexp =\Omega\ab\Big(\frac{n^2\Lloc^2}{\eps^2}).
\end{equation}
This bound holds even for adaptive algorithms with arbitrary ancillas, entangled inputs, quantum memory, and knowledge of the support.
\end{theorem}
The hard instances $\cL_{\pm2\eps}$ from~\cref{eq:lb_family} are $3$-local and have identical Hamiltonian and dissipative supports. Thus, the lower bound already holds for $3$-local Lindbladians with known support.

\begin{proof}
Consider the two instances $\cL_{\pm h}$ with $m=n-1$, $r=\Lloc/16$, and $h=2\eps$. Then $h\leq r/64$, and \cref{lem:lb_strength} gives locality and the strength promise. An $\eps$-accurate estimation of the Hamiltonian coefficient at $Z_0$ distinguishes $h$ from $-h$ with error at most $1/3$. By~\cref{lem:lb_uniform_overlap}, the two channels admit Stinespring isometries with
\[
V_{+,t}^\dag V_{-,t}=s(t)I,\qquad 1-\beta\leq s(t)\leq1,\qquad \beta=\frac{2^{19}\eps^2}{(n-1)(n-2)\Lloc^2},
\]
uniformly over $t\geq0$.

We track the inner product of the two purifications.
Pad each terminated branch with zero-time calls so that every branch has exactly $\Nexp$ rounds.
Purify the initial state and all operations, use $V_{\pm,t}$ with a fresh inaccessible environment for each call, and store classical histories as orthogonal records indexed by finite bit strings.
Before the $j$th call, $j\in[\Nexp]$, let $\ket{\psi_{d,jb}}$ be the subnormalized branch state under hypothesis $d\in\{+,-\}$ and history $b$, with probabilities and overlaps
\[
p_d(jb)=\norm{\ket{\psi_{d,jb}}}_2^2,\qquad a_{jb}=\braket{\psi_{+,jb}}{\psi_{-,jb}}.
\]
For each round $j$, Cauchy--Schwarz gives
\[
|a_{jb}|\leq\sqrt{p_+(jb)p_-(jb)}\leq\frac{p_+(jb)+p_-(jb)}2.
\]
In particular, $\sum_b|a_{jb}|\leq1$ because $\sum_b p_d(jb)=1$.

Let $A_{\mathrm f}$ be the inner product of the normalized final purifications, including their orthogonal history records. Common purified operations preserve the sum of branch overlaps. For measurements, this follows from $\sum_cM_c^\dag M_c=I$. At a call on history $b$, both hypotheses use the same duration $t_{jb}$, and $V_{+,t_{jb}}^\dag V_{-,t_{jb}}=s(t_{jb})I$ changes the branch overlap from $a_{jb}$ to $s(t_{jb})a_{jb}$. Summing these changes over the $\Nexp$ rounds gives
\[
1-A_{\mathrm f}=\sum_{j=1}^{\Nexp}\sum_b(1-s(t_{jb}))a_{jb}.
\]
Since $0\leq1-s(t)\leq\beta$, we obtain
\begin{align}
|1-A_{\mathrm f}| &\leq\sum_{j=1}^{\Nexp}\sum_b(1-s(t_{jb}))|a_{jb}|\nonumber\\
&\leq\beta\sum_{j=1}^{\Nexp}\sum_b|a_{jb}|\leq\beta\Nexp. \label{eq:lb_count_progress}
\end{align}

The original protocol distinguishes the two instances with error at most $1/3$ on each. Its final accessible states $\rho_\pm$, including the decision, therefore satisfy
\[
\frac13\leq\frac12\norm{\rho_+-\rho_-}_1\leq\sqrt{1-|A_{\mathrm f}|^2},
\]
where the second inequality follows by tracing out the purification registers. Hence $|A_{\mathrm f}|\leq\sqrt{8/9}$ and
\[
1-\sqrt{8/9}\leq1-|A_{\mathrm f}|\leq|1-A_{\mathrm f}|\leq\beta\Nexp.
\]
Substituting the value of $\beta$ gives
\[
\Nexp\geq\frac{1-\sqrt{8/9}}{\beta}
=\frac{1-\sqrt{8/9}}{2^{19}}\frac{(n-1)(n-2)\Lloc^2}{\eps^2},
\]
which proves~\cref{eq:lb_count_theorem}.
\end{proof}

\subsection{Lower bound on expected total evolution time}\label{sec:lb_time}

We now use the bound $1-s(t)\leq\kappa t$ from \cref{lem:lb_uniform_overlap} to bound the expected total evolution time.

\begin{theorem}[Lower bound on expected total evolution time]\label{thm:lb_time}
Let $k\geq3$, $\Lloc>0$, and $0<\eps\leq\Lloc/2048$. Suppose an algorithm estimates the Hamiltonian coefficients of every $k$-local $n$-qubit Lindbladian with local dynamical strength at most $\Lloc$ to entrywise error $\eps$, with success probability at least $2/3$. Then
\begin{equation}\label{eq:lb_time_theorem}
\sup_{\cL}\E_{\cL}\Ttot =\Omega\ab\Big(\frac{n\Lloc}{\eps^2}).
\end{equation}
This bound holds even for adaptive algorithms with arbitrary ancillas, entangled inputs, quantum memory, and knowledge of the support.
\end{theorem}
The same two $3$-local hard instances $\cL_{\pm2\eps}$ from~\cref{eq:lb_family} have identical Hamiltonian and dissipative supports. Hence, this lower bound also holds for $3$-local Lindbladians with known support.

\begin{proof}
Consider the two instances $\cL_{\pm h}$ with $m=n-1$, $r=\Lloc/16$, and $h=2\eps\leq r/64$. By~\cref{lem:lb_strength}, these are valid instances, and an $\eps$-accurate estimate of their Hamiltonian coefficient at $Z_0$ distinguishes $h$ from $-h$ with error at most $1/3$ on each.

By~\cref{lem:lb_uniform_overlap}, use the same Stinespring isometries as in the experiment count proof, with
\begin{equation}\label{eq:lb_time_overlap}
V_{+,t}^\dag V_{-,t}=s(t)I,\qquad 0\leq1-s(t)\leq\kappa t,\qquad \kappa=\frac{768\eps^2}{(n-1)\Lloc}.
\end{equation}

Let $M_T=\max_{d\in\{+,-\}}\E_d\Ttot$. If $M_T=\infty$, the claimed bound is immediate, so assume $M_T<\infty$. We first consider protocols with a finite worst-case query budget $\Nexp$. Purify and pad the protocol as in the proof of~\cref{thm:lb_count}, using $V_{\pm,t}$ with a fresh reference environment for each call and keeping classical histories in orthogonal records. Before the $j$th call, $j\in[\Nexp]$, on history $b$, let $\ket{\psi_{d,jb}}$ be the subnormalized branch states, with probabilities and overlaps
\[
p_d(jb)=\norm{\ket{\psi_{d,jb}}}_2^2,\qquad a_{jb}=\braket{\psi_{+,jb}}{\psi_{-,jb}}.
\]
Writing $t_{jb}$ for the duration chosen on history $b$, we have
\begin{align*}
\sum_{j=1}^{\Nexp}\sum_b t_{jb}p_d(jb)=\E_d\sum_{j=1}^{\Nexp} t_j=\E_d\Ttot,\\
|a_{jb}|\leq\sqrt{p_+(jb)p_-(jb)}\leq\frac{p_+(jb)+p_-(jb)}2.
\end{align*}
Let $A_{\mathrm f}$ be the inner product of the normalized final purifications, including their orthogonal history records.
As in the proof of~\cref{thm:lb_count}, summing the changes in branch overlaps gives
\[
1-A_{\mathrm f}=\sum_{j=1}^{\Nexp}\sum_b(1-s(t_{jb}))a_{jb}.
\]
Using~\cref{eq:lb_time_overlap}, it follows that
\begin{align}
|1-A_{\mathrm f}| &\leq\sum_{j=1}^{\Nexp}\sum_b(1-s(t_{jb}))|a_{jb}|\nonumber\\
&\leq\frac{\kappa}{2}\sum_{j=1}^{\Nexp}\sum_b t_{jb} \ab\big(p_+(jb)+p_-(jb))\nonumber\\
&=\frac{\kappa}{2}\ab\big(\E_+\Ttot+\E_-\Ttot) \leq\kappa M_T. \label{eq:lb_time_progress}
\end{align}

The original protocol distinguishes the two instances with error at most $1/3$ on each. Its final accessible states $\rho_\pm$, including the decision, therefore satisfy
\[
\frac13\leq\frac12\norm{\rho_+-\rho_-}_1\leq\sqrt{1-|A_{\mathrm f}|^2},
\]
where the second inequality follows by tracing out the purification registers. Hence $|A_{\mathrm f}|\leq\sqrt{8/9}$, and~\cref{eq:lb_time_progress} gives
\[
1-\sqrt{8/9}\leq1-|A_{\mathrm f}|\leq|1-A_{\mathrm f}|\leq\kappa M_T.
\]
Dividing by $\kappa$ and substituting its value gives
\[
\sup_{\cL}\E_{\cL}\Ttot\geq M_T\geq\frac{1-\sqrt{8/9}}{\kappa}
=\frac{1-\sqrt{8/9}}{768}\frac{(n-1)\Lloc}{\eps^2},
\]
which proves~\cref{eq:lb_time_theorem} for protocols with a finite worst-case query budget.

For protocols with a random number of queries and no finite worst-case query bound, truncate the protocol just before its $(N+1)$st query, retaining any decision already returned and assigning an arbitrary decision to unfinished branches. Let $p_{d,N}$ be the probability that the original protocol returns the correct decision after at most $N$ queries under hypothesis $d\in\{+,-\}$. These probabilities increase to the original success probabilities, so $\lim_{N\to\infty}p_{d,N}\geq2/3$. The truncated protocol succeeds with probability at least $p_{d,N}$ under each hypothesis, and its expected total evolution time is at most $M_T$. Its final accessible states therefore have trace distance
\[
D_N\coloneq\frac12\norm{\rho_{+,N}-\rho_{-,N}}_1\geq p_{+,N}+p_{-,N}-1,
\qquad \liminf_{N\to\infty}D_N\geq\frac13.
\]
Applying the finite-round overlap argument above to the truncated protocol gives a final purification overlap $A_N$ satisfying $|1-A_N|\leq\kappa M_T$ and $D_N\leq\sqrt{1-|A_N|^2}$. Consequently,
\[
1-\sqrt{1-D_N^2}\leq1-|A_N|\leq|1-A_N|\leq\kappa M_T.
\]
Taking the lower limit as $N\to\infty$ yields $1-\sqrt{8/9}\leq\kappa M_T$, proving the same lower bound without a finite worst-case query budget.
\end{proof}

By~\cref{lem:lb_strength}, the global dynamical strength of the hard instances is of order $n$ times their local dynamical strength. This yields the following lower bounds under a global strength bound.

\begin{corollary}[Lower bounds in global strength] \label{cor:lb_global}
Let $k\geq3$, $\Lambda>0$, and $0<\eps\leq\Lambda/(2048n)$. For learning $k$-local $n$-qubit Lindbladians under the global promise $\norm{\cL^\dag}_{\infty\to\infty}\leq\Lambda$, the worst-case experiment budget and the worst-case expected total evolution time obey
\[
\Nexp =\Omega\ab\Big(\frac{\Lambda^2}{\eps^2}), \qquad \sup_{\cL}\E_{\cL}\Ttot =\Omega\ab\Big(\frac{\Lambda}{\eps^2}).
\]
These bounds already hold for $3$-local instances whose actual global strength is between $\Lambda/8$ and $\Lambda$.
\end{corollary}

\begin{proof}
Set $\Lloc=\Lambda/n$ in~\cref{thm:lb_count,thm:lb_time}. By~\cref{lem:lb_strength}, the two hard instances satisfy the global promise and the asserted lower bound on their actual strength. Substitution gives the two resource bounds because $(n-1)(n-2)/n^2\geq2/9$ and $(n-1)/n\geq2/3$.
\end{proof}

\section*{AI use disclosure}
\addcontentsline{toc}{section}{AI use disclosure}
\begingroup
\setlength{\emergencystretch}{1em}
GPT-6 Astra was used interactively to assist in refining the hard-instance construction in Section~\ref{sec:lb_family} and in developing and checking the proof of Lemma~\ref{lem:lb_uniform_overlap}.
GPT-5.6 Sol was used to check the correctness of proofs, identify relevant references, and polish the manuscript.
The authors take full responsibility for the final text and all mathematical statements, arguments, and results.
\par
\endgroup

\section*{Acknowledgments}
\addcontentsline{toc}{section}{Acknowledgments}
Z.C.\ and Z.Y.\ acknowledge support by the CQT Young Researcher Career Development Grant.

\bibliographystyle{alphaurl}
\bibliography{ref}
\end{document}